\documentclass[sigconf]{acmart}

\AtBeginDocument{%
  \providecommand\BibTeX{{%
    \normalfont B\kern-0.5em{\scshape i\kern-0.25em b}\kern-0.8em\TeX}}}

\setcopyright{acmcopyright}
\copyrightyear{2018}
\acmYear{2018}
\acmDOI{10.1145/1122445.1122456}

\newcommand{\Order}{\mathrm{O}}
\newcommand{\OrderT}{\tilde{\mathrm{O}}}
\newcommand{\OmegaT}{\tilde{\mathrm{\Omega}}}
\newcommand{\ThetaT}{\tilde{\Theta}}

\newcommand{\poly}{\mathop{\mathrm{poly}}\nolimits}

\usepackage{enumitem}
\usepackage{color}
\usepackage{graphicx}
\usepackage{microtype}
\usepackage{tikz}
\usepackage{algorithm}
\usepackage[noend]{algorithmic}

\usetikzlibrary{arrows,matrix,positioning,patterns}

\newtheorem{problem}{Problem}

\newcommand{\Exp}{\mathop{\mathbb{E}}}
\renewcommand{\Pr}{\mathbb{P}}

\begin{document}

%%
%% The "title" command has an optional parameter,
%% allowing the author to define a "short title" to be used in page headers.
\title{Frequent Elements with Witnesses in Data Streams}

%%
%% The "author" command and its associated commands are used to define
%% the authors and their affiliations.
%% Of note is the shared affiliation of the first two authors, and the
%% "authornote" and "authornotemark" commands
%% used to denote shared contribution to the research.
\author{Christian Konrad}
%\authornote{TODO: Author note}
\email{christian.konrad@bristol.ac.uk}
\orcid{0000-0003-1802-4011}
\affiliation{%
  \institution{Department of Computer Science, University of Bristol}
  \city{Bristol}
  \country{United Kingdom}
}

%%
%% By default, the full list of authors will be used in the page
%% headers. Often, this list is too long, and will overlap
%% other information printed in the page headers. This command allows
%% the author to define a more concise list
%% of authors' names for this purpose.
%\renewcommand{\shortauthors}{Trovato and Tobin, et al.}

%%
%% The abstract is a short summary of the work to be presented in the
%% article.
\begin{abstract}
  Detecting frequent elements is among the oldest and most-studied problems in the area of data streams. Given a stream of $m$ data items in $\{1, 2, \dots, n\}$, the objective is to output items that appear at least $d$ times, for some threshold parameter $d$, and provably optimal algorithms are known today. However, in many applications, knowing only the frequent elements themselves is not enough: For example, an Internet router may not only need to know the most frequent destination IP addresses of forwarded packages, but also the timestamps of when these packages appeared or any other meta-data that ``arrived'' with the packages, e.g., their source IP addresses.
  
  In this paper, we introduce the witness version of the frequent elements problem: Given a desired approximation guarantee $\alpha \ge 1$ and a desired frequency $d \le \Delta$, where $\Delta$ is the frequency of the most frequent item, the objective is to report an item together with at least $d / \alpha$ timestamps of when the item appeared in the stream (or any other meta-data that arrived with the items). We give provably optimal algorithms for both the insertion-only and insertion-deletion stream settings:  In insertion-only streams, we show that space $\OrderT(n + d \cdot n^{\frac{1}{\alpha}})$ is necessary and sufficient for every integral $1 \le \alpha \le  \log n$.  In insertion-deletion streams, we show that space $\OrderT(\frac{n \cdot d}{\alpha^2})$ is necessary and sufficient, for every $\alpha \le \sqrt{n}$.
  
\end{abstract}

%%
%% The code below is generated by the tool at http://dl.acm.org/ccs.cfm.
%% Please copy and paste the code instead of the example below.
%%

%%
%% Keywords. The author(s) should pick words that accurately describe
%% the work being presented. Separate the keywords with commas.

%%
%% This command processes the author and affiliation and title
%% information and builds the first part of the formatted document.
\maketitle

\section{Introduction}
 The {\em streaming model of computation} addresses the fundamental issue that modern massive data sets 
 are too large to fit into the Random-Access Memory (RAM) of modern computers. Typical examples of such 
 data sets are Internet traffic logs, financial transaction streams, and massive graphical
 data sets, such as the Web graph and social network graphs. A data streaming algorithm receives its input 
 piece by piece in a linear fashion and has access to only a sublinear amount of memory. This prevents 
 the algorithm from seeing the input in its entirety at any one moment.

 The \textsf{Frequent Elements} (\textsf{FE}) (or heavy hitters) problem is among the oldest and most-studied problems in the area of data streams. Given a stream 
 $S=s_1, s_2, \dots, s_m$ of length $m$ with $s_i \in [n]$, for some integer $n$, the goal is to identify items in $[n]$ that appear at least $d = \epsilon m$ times, for some $\epsilon > 0$. This problem was first solved by Misra and Gries in 1982 \cite{mg82} and has since been addressed in countless research papers (e.g. \cite{ccf02,mm02,ev03,mae05,cm05,kx06,bcis09,cfm09,bdw19}), culminating in provably optimal algorithms
 %that use space $\Order(\frac{1}{\epsilon} + \log n + \log \log m)$ 
 \cite{bdw19}. 
 
 However, in many applications, only knowing the frequent items themselves is insufficient, and additional application-specific
 data is required. For example:
 
 \begin{itemize}
  \item Given a database log, a \textsf{FE} algorithm can be used to detect a frequently updated (or queried)
 entry. However, users who committed these updates (or queries) or the timestamps of when these updates
 (or queries) were executed cannot be reported by such an algorithm. 
 
  \item Given a stream of friendship updates in a social network graph, a \textsf{FE} algorithm can detect nodes of large 
 degree (e.g., an influencer in a social network). Their neighbours (e.g., followers of an influencer), however, cannot 
 be outputted by such an algorithm. 
  \item Given the traffic log of an Internet router logging timestamps, source, and destination IP addresses of forwarded IP packages,
  Denial-of-Service attacks can be detected by identifying {\em distinct frequent elements}, that is, frequent target IP addresses 
  that are requested from many distinct sources \cite{fabcs17}. Here, a (distinct) \textsf{FE} algorithm only reports frequent target IP addresses 
 and thus potential machines that were under attack, however, the timestamps of when these attacks occurred or the 
 source IP addresses from where the attacks originated remain unknown.  
 \end{itemize}

 In this paper, we introduce the witness version of the frequent elements problem, which captures the examples mentioned above. This problem is formulated as a problem on graphs:
 
 \begin{problem}[Frequent Elements with Witnesses (\textsf{FEwW}]
 In \textsf{FEwW}$(n, d)$, the input consists of a bipartite graph $G = (A, B, E)$ with $|A| = n$ and 
 $|B| = m = \poly n$, and a threshold parameter $d$. We are given the promise that there is at least one $A$-vertex of degree at least $d$. %A data streaming algorithm receives a sequence of the edges of $G$ in arbitrary order. 
 The goal is to output an $A$-vertex together with at least $d/\alpha$ of its neighbours, for some approximation factor $\alpha \ge 1$. 
\end{problem}
\textsf{FEwW} allows us to model frequent elements problems where, besides the frequent elements themselves, additional satellite data that ``arrives'' together with the input items also needs to be reported. For example, in the 
 database log example above, database entries 
 can be regarded as $A$-vertices, users as $B$-vertices, and updates/queries as edges connecting entries to users. The incident edges of a node reported by an algorithm for \textsf{FEwW} can be regarded as a ``witness'' that proves that the node is indeed of large degree. The restriction $|B| = \poly n$ is only imposed for convenience as it is reasonable and simplifies the complexity bounds of our algorithms.
 
Our aim is to solve \textsf{FEwW} in two models of graph streams: In the \emph{insertion-only} model, a streaming algorithm receives the edges of $G$ in arbitrary order. In the \emph{insertion-deletion} model, a streaming algorithm receives an arbitrary sequence of edge insertion and deletions. In both models, the objective is to design algorithms with minimal space.

 Formulating \textsf{FEwW} as a graph problem has two advantages: First, it allows for the same satellite data of different input items. Second, a streaming algorithm for \textsf{FEwW} can be used to solve the related \textsf{Star Detection} problem, a subgraph detection problem that deserves attention in its own right:
 \begin{problem}[\textsf{Star Detection}]
 In \textsf{Star Detection}, the input is a general graph $G=(V, E)$. The objective is to output the largest star in $G$, i.e., determining a node of largest degree together with its neighbourhood. An $\alpha$-approximation algorithm ($\alpha \ge 1$) to 
 \textsf{Star Detection} outputs a node together with at least $\Delta / \alpha$ of its neighbours, where $\Delta$ is the maximum degree in the input graph. 
\end{problem}
For example, \textsf{Star Detection} can be used to solve the second example mentioned above, i.e., finding
influencers together with their followers in social networks.  %This problem ties in with various works in the streaming literature that address the problem 
%of approximating certain graph structures, such as large independent sets and cliques \cite{hssw12,cdk19}, and 
%large matchings (e.g. \cite{fkmsz05,kmm12}).
 
\subsection{Our Results}
In this paper, we resolve the space complexity of streaming algorithms for \textsf{FEwW} in both insertion-only and insertion-deletion streams up to poly-logarithmic factors. 
 
In insertion-only streams, we give an $\alpha$-approximation streaming algorithm with space\footnote{We use $\tilde{O}$, $\tilde{\Theta}$, and $\tilde{\Omega}$ to mean $O$, $\Theta$, and $\Omega$ (respectively) with poly-log factors suppressed.} $\OrderT(n + n^{\frac{1}{\alpha}} d)$ that succeeds 
 with high probability\footnote{We say that an event occurs with high probability (in short: w.h.p.) if it happens with probability at least 
 $1-\frac{1}{n}$, where $n$ is a suitable parameter associated with the input size.}, for integral values of $\alpha \ge 1$ (\textbf{Theorem~\ref{thm:insertion-only-upper-bound}}). 
 %This algorithm can also be used to obtain a $\Order(\log n)$-approximation semi-streaming algorithm
 %for \textsf{Star Detection} (\textbf{Corollary~\ref{cor:star-detection}}). 
 We complement this result with a lower bound, showing 
 that space $\Omega(n/ \alpha^2 + n^{\frac{1}{\alpha-1}} d / \alpha^2 )$ is necessary for every algorithm that computes a 
 $\alpha / 1.01$ approximation, for every integer $\alpha \ge 2$ (\textbf{Theorems~\ref{thm:lb-set-disjointness} and \ref{thm:insertion-only-lower-bound}}). 
 Observe that the latter result also implies a lower bound of $\Omega(n/ \alpha^2 + n^{\frac{1}{\alpha}} d / \alpha^2 )$ for every $\alpha$-approximation
 algorithm, where $\alpha$ is integral.
 Up to poly-logarithmic factors, our algorithm is thus optimal for every poly-logarithmic $\alpha$.
 
 In insertion-deletion streams, we give an $\alpha$-approximation streaming algorithm with space 
 $\OrderT(\frac{dn}{\alpha^2})$ if $\alpha \le \sqrt{n}$, and space $\OrderT(\frac{\sqrt{n}d}{\alpha})$ if $\alpha > \sqrt{n}$ that succeeds w.h.p.
 (\textbf{Theorem~\ref{thm:insertion-deletion-upper-bound}}). %This result yields a $\Order(\sqrt{n})$-approximation
 %semi-streaming algorithm for \textsf{Star Detection} (\textbf{Corollary~\ref{cor:insertion-deletion-star-detection}}).
 We complement our algorithm with a lower bound showing that space 
 $\OmegaT(\frac{dn}{\alpha^2})$ is required (\textbf{Theorem~\ref{thm:insertion-deletion-lower-bound}}), 
 which renders our algorithm optimal (if $\alpha \le \sqrt{n}$) up to poly-logarithmic factors. 
 
% Our results are summarized in Table~\ref{tab:results}.
 
%  \begin{table}[h!]
%  \begin{center}\begin{tabular}{l|l|l}
%   \textsf{Neighborhood Detection} & \textbf{Space bound} & \textbf{Result} \\
%   \hline 
%   Insertion-only streams & $\OrderT(n + n^{\frac{1}{c}} d)$, for integers $c \ge 1$ & Theorem~\ref{thm:insertion-only-upper-bound} \\ 
%    & $\Omega(n/ c^2 + n^{\frac{1}{c}} d / c^2 )$ & Theorems~\ref{thm:insertion-only-lower-bound} and \ref{thm:lb-set-disjointness} \\
%    \hline
%   Insertion-deletion streams & $\ThetaT(\frac{dn}{c^2})$, for $c \le \sqrt{n}$ & Theorems~\ref{thm:insertion-deletion-upper-bound} (alg.) and \ref{thm:insertion-deletion-lower-bound} (LB) \\
%  & $\OrderT(\frac{\sqrt{n}d}{c})$ if $c > \sqrt{n}$ & Theorem~\ref{thm:insertion-deletion-upper-bound}
%  \end{tabular} \end{center}
%  \caption{Summary of our results regarding the space requirements of one-pass $c$-approximation streaming algorithms for \textsf{Neighborhood Detection}. \label{tab:results}}
%  \end{table}

% Our algorithm for insertion-only streams (\textbf{Theorem~\ref{thm:insertion-only-upper-bound}}) combined with our lower bound
% for insertion-deletion streams (\textbf{Theorem~\ref{thm:insertion-deletion-lower-bound}}) establish 
% a separation result between the insertion-only and the insertion-deletion models.
 
 Our lower bounds translate to \textsf{Star Detection} with parameter $d = \Theta(n)$, and 
 our algorithms translate to \textsf{Star Detection} by setting $d = \Theta(n)$ in the space bound and by introducing 
 an additional $\log_{1+\epsilon} n$ factor in the space complexities 
 and a $1+\epsilon$ factor in the approximation ratios (\textbf{Lemma~\ref{lem:neighbourhood-detection-to-star-detection}}). For example, a $\Order(\log n)$-approximation to \textsf{Star Detection}
 can be computed in insertion-only streams in space $\OrderT(n)$ (graph streaming algorithms with space $\OrderT(n)$ are referred to as {\em semi-streaming algorithms} \cite{fkmsz05}), while such an approximation would require space $\OmegaT(n^2)$ in insertion-deletion streams. 
 
\subsection{Techniques}
Our insertion-only streaming algorithm for \textsf{FEwW} makes use of a subroutine that solves the following sampling task: For degree bounds $d_1 < d_2$ and an integer $s$, compute a uniform random sample of size $s$ of the $A$-vertices of degree at least $d_1$, and, for every sampled vertex $a \in A$, compute $\min \{ d_2, \deg(a) - d_1 + 1 \}$ incident edges to $a$. We say that this task \emph{succeeds} if there is one sampled node for which $d_2$ incident edges are computed. We give a streaming algorithm, denoted \textsc{Deg-Res-Sampling}$(d_1, d_2, s)$, that solves this task, using a combination of reservoir sampling \cite{v85} and degree counts. Next, we run $\alpha$ instances  of \textsc{Deg-Res-Sampling}$(d_1, d_2, s)$ in parallel, for changing parameter $d_1 = i \cdot \frac{d}{\alpha}$, for $i = 0, 1, \dots, \alpha -1$, and fixed parameters $d_2 = \frac{d}{\alpha}$ and $s = \ThetaT(n^{1/\alpha})$. It can be seen that run $i$ succeeds if the ratio of the number of nodes of degree at least $i \cdot \frac{d}{\alpha}$ to the number of nodes of degree at least $(i+1) \cdot \frac{d}{\alpha}$ in the input graph is not too large, i.e., in $\Order(n^{1/\alpha})$. We prove that this condition is necessarily fulfilled for at least one of the parallel runs.

Our lower bound for insertion-only streams is the most technical contribution of this paper. We show that a streaming algorithm for \textsf{FEwW} can be used to solve a new multi-party one-way communication problem denoted \textsf{Bit-Vector-Learning}, where the bits of multiple binary strings of different lengths are partitioned among multiple parties. The last party is required to output enough bits of at least one of the strings - this is difficult, since the partitioning is done so that not a single party alone holds enough bits of any of the strings. We prove a lower bound on the communication complexity of \textsf{Bit-Vector-Learning} via information theoretic arguments, which then translates to \textsf{FEwW}. A highlight of our technique is the application of Baranyai's theorem for colouring complete regular hypergraphs \cite{b79}, which allows us to partition and subsequently quantify the information that is necessarily revealed when solving \textsf{Bit-Vector-Learning}. 

\textsf{FEwW} is much harder to solve in insertion-deletion streams and requires a different set of techniques. Our insertion-deletion streaming algorithm employs two sampling strategies: A vertex-based sampling strategy that succeeds if the input graph is dense enough, and an edge sampling strategy that succeeds if the input graph is relatively sparse. We implement both sampling methods using $l_0$-sampling techniques \cite{jst11}. 

Last, our lower bound for insertion-deletion streams is proved in the one-way two-party communication model and is conceptually interesting since it extends the traditional one-way two-party \textsf{Augmented-Index} communication problem to a suitable two dimensional version that may be of independent interest. Similar to our lower bound for insertion-only streams, we use information-theoretic arguments to prove a tight lower bound.

\subsection{Further Related Work} 
As previousy mentioned, the traditional (without witnesses) \textsf{FE} problem is very well studied, and many algorithms with different properties are known, including Misra-Gries \cite{mg82} (see also \cite{dlm02,ksp03}), CountSketch \cite{ccf04}, Count-min Sketch \cite{cm05}, multi-stage Bloom filters \cite{cfm09}, and many others \cite{mm02,ev03,maa05,kx06,bics10}. A crucial difference between the witness and the without witness versions is that the space complexities behave inversely with regards to the desired frequency threshold parameter $1 \le d \le m$: Most streaming algorithms for \textsf{FE} use space proportional to $\frac{m}{d}$ (intuitively, the more often an element appears in the stream, the easier it is to pick it up using sampling), while the space is trivially $\Omega(d / \alpha)$ for \textsf{FEwW}, since at least $d / \alpha$ witnesses need to be reported by the algorithm. In terms of techniques, the two versions therefore have a very different flavour, and the \textsf{FEwW} problem is perhaps closer in spirit to the literature on graph streaming algorithms than to the (without witnesses) frequent elements literature.

Graph streaming algorithms in the insertion-only model have been studied since more than 20 years \cite{hrr99}, and this model is fairly well understood today (see \cite{m14} for an excellent survey). The first techniques for processing insertion-deletion graph streams were introduced in a seminal paper by Ahn et al. \cite{agm12} in 2012. While many problems, such as \textsf{Connectivity} \cite{agm12}, \textsf{Spectral Sparsification} \cite{kmmmnst20}, and \textsf{$(\Delta+1)$-colouring} \cite{ack19}, are known to be equally hard in both the insertion-only and the insertion-deletion settings (up to a poly-logarithmic factor difference in the space requirements), only few problems, such as \textsf{Maximum Matching} and \textsf{Minimum Vertex Cover}, are known to be substantially harder in the insertion-deletion setting \cite{k15,akly16,dk20}. In this paper, we prove that \textsf{FEwW} and \textsf{Star Detection} are much harder in the insertion-deletion setting than in the insertion-only setting, thereby establishing another separation result between the two models.

\textsf{Star Detection} shares similarities with other subgraph approximation problems, such as \textsf{Maximum Independent Set}/\textsf{Maximum Clique} \cite{hssw12,cdk19}, \textsf{Maximum Matching} \cite{fkmsz05,kmm12,k15,akly16,dk20}, and \textsf{Minimum Vertex Cover} \cite{dk20}, which can all be solved approximately using sublinear (in $n^2$) space in both the insertion-only and insertion-deletion settings. 
%At a first glace, other subgraph detection problems, such as deciding whether or not the input graph contains a triangle, seem closely related, however, the fact that the structure of a triangle cannot be approximated renders these problems different.  

\subsection{Outline}
We start with notations and definitions in Section~\ref{sec:preliminaries}. This section also introduces the necessary context on communication complexity needed in this work.
%From a technical perspective, we consider our
%lower bounds as the technically most interesting and challenging contributions. 
In Section~\ref{sec:alg-insertion-only}, we give our insertion-only streams algorithm, and in 
Section~\ref{sec:lb-insertion-only}, we present a matching lower bound. Our algorithm for 
insertion-deletion streams is given in Section~\ref{sec:alg-insertion-deletion}, and we conclude 
with a matching lower bound in Section~\ref{sec:lb-insertion-deletion}.

\section{Preliminaries} \label{sec:preliminaries}
We consider simple bipartite graphs $G = (A, B, E)$ with $|A| = n$
and $|B| = m = \poly(n)$. The maximum degree of an $A$-node is denoted by $\Delta$.
We say that a tuple $(a, S) \in A \times 2^B$ is a {\em neighbourhood} in $G$ if 
$S \subseteq \Gamma(a)$. The size $|(a, S)|$ of $(a, S)$ is defined as $|(a, S)| = |S|$.
%Observe that $|(a, S)| \le \Delta$, for any neighbourhood $(a, S)$. 
Using this terminology, the objective of \textsf{FEwW} is to output a neighbourhood of size at least $d/\alpha$.

Let $A$ be a random variable distributed according to $\mathcal{D}$. 
The {\em Shannon Entropy} of $A$ is denoted by $H_{\mathcal{D}}(A)$,
or simply $H(A)$ if the distribution $\mathcal{D}$ is clear from the context. The {\em mutual information}
of two jointly distributed random variables $A, B$ with distribution $\mathcal{D}$ is denoted by 
$I_{\mathcal{D}}(A, B) := H_{\mathcal{D}}(A) - H_{\mathcal{D}}(A \ | \ B)$ 
(again, $\mathcal{D}$ may be dropped),
where $H_{\mathcal{D}}(A \ | \ B)$ is the entropy of $A$ conditioned on $B$.
For an overview on information theory we refer the reader to \cite{ct06}.

\subsection*{Communication Complexity} \label{sec:communication-complexity}
We now provide the necessary context on communication complexity %and its connections to lower bounds for data streaming algorithms
(see \cite{kn06} for more information).

In the \textit{one-way $p$-party communication model}, for $p \ge 2$, $p$ parties $P_1, P_2, \dots, P_p$ communicate with each
other to jointly solve a problem. Each party $P_i$ holds their own private input $X_i$ and has access to both private
and public random coins. Communication is one-way: $P_1$ sends a message $M_1$ to $P_2$, who then sends a message 
$M_2$ to $P_3$. This process continues until $P_p$ receives a message $M_{p-1}$ from $P_{p-1}$ and then 
outputs the result.

%, where the $(X_i)_{i \in [p]}$ 
%are jointly distributed random variables according to some distribution $\mathcal{D}$. In addition to their inputs, each party
%$P_i$ has access to both private and public random bits. 

The way the parties interact is specified by a communication protocol $\Pi$. We say that $\Pi$ is an $\epsilon$-error
protocol for a problem $\mathsf{Prob}$ if it is correct with probability $1-\epsilon$ on any input 
$(X_1, X_2, \dots, X_p)$ that is valid for $\mathsf{Prob}$, where the probability is taken over the randomness
(both private and public) used by the protocol. 
The \textit{communication cost} of $\Pi$ is the size of the longest message sent by any of the parties, 
that is, $\max_{1 \le i \le p - 1} \{|M_i| \}$, where $|M_i|$ is the maximum length of message 
$M_i$. The {\em randomized one-way communication complexity} $R_{\epsilon}^{\rightarrow}(\mathsf{Prob})$ of a problem 
$\textsf{Prob}$ is the minimum communication cost among all $\epsilon$-error protocols $\Pi$.

Let $\mathcal{D}$ be any input distribution for a specific problem $\mathsf{Prob}$. The 
{\em distributional one-way communication complexity} of $\mathsf{Prob}$, denoted $D_{\mathcal{D}, \epsilon}^{\rightarrow}(\mathsf{Prob})$,
is the minimum communication cost among all deterministic communication protocols for $\mathsf{Prob}$ that succeed
with probability at least $1-\epsilon$, where the probability is taken over the input distribution $\mathcal{D}$.
In order to prove lower bounds on $R_{\epsilon}^{\rightarrow}(\mathsf{Prob})$, by Yao's lemma it is enough to 
bound the distributional communication complexity for any suitable input distribution since
$R_{\epsilon}^{\rightarrow}(\mathsf{Prob}) = \max_{\mathcal{D}} D_{\mathcal{D}, \epsilon}^{\rightarrow}(\mathsf{Prob})$.
In our lower bound arguments we will therefore consider deterministic protocols with distributional error. 
This is mainly for convenience as this allows us to disregard public and private coins. We note, however, 
that with additional care about private and public coins, our arguments also directly apply to randomized protocols.

Our lower bound arguments follow the {\em information complexity} paradigm. There are various definitions 
of information complexity (e.g. \cite{bjks02,bjks022,cswy01}), and for the sake of simplicity we will in fact omit a precise definition. 
Information complexity arguments typically measure the amount 
of information revealed by a communication protocol about the inputs of the participating parties.
This quantity is a natural lower bound on the total amount of communication, since the amount of information revealed
cannot exceed the number of bits exchanged. We will follow this approach in that we give lower bounds on quantities of the 
form $I_{\mathcal{D}}(X_i \ : \ M_j)$, for some $j \ge i+1$. This then implies a lower bound on the communication complexity
of a specific problem $\mathsf{Prob}$ since $I_{\mathcal{D}}(X_i \ : \ M_j) \le H_{\mathcal{D}}(M_j) \le |M_j|$ holds 
for any protocol. 
%For the sake of simplicity, we omit defining information complexity with regards to our problems.

%Last, to obtain space lower bounds for data streaming algorithms, it is enough to show that any data streaming algorithm 
%$A$ with space $s$ for a specific problem $\textsf{P}$ yields a communication protocol with communication 
%cost $f(s)$ for a problem $\textsf{Q}$, for some function $f$ and suitable problems $\textsf{P}, \textsf{Q}$. A lower bound 
%on the communication complexity of $\textsf{Q}$ thus yields a space lower bound for $A$.

\section{Algorithm for Insertion-only Streams} \label{sec:alg-insertion-only}
Before presenting our algorithm for \textsf{FEwW} in insertion-only streams, we discuss 
a sampling subroutine that combines reservoir sampling with degree counts. 

\subsection{Degree-based Reservoir Sampling}
The subroutine \textsc{Deg-Res-Sampling}$(d_1, d_2, s)$ samples $s$ nodes uniformly at random from the set of 
nodes of degree at least $d_1$, and for each of these nodes computes a neighbourhood of size 
$\min \{ d_2, \deg - d_1 + 1 \}$, where $\deg$ is the degree of the respective node. If at least one neighbourhood
of size $d_2$ is found then we say that the algorithm {\em succeeds} and returns an arbitrary neighbourhood 
among the stored neighbourhoods of sizes $d_2$. Otherwise, we say that the algorithm {\em fails} and it reports \texttt{fail}.

This is achieved as follows: While processing the stream of edges, the degrees
of all $A$-vertices are maintained. The algorithm maintains a reservoir of size $s$ that fulfils the invariant that, 
at any moment, it contains a uniform sample of size $s$ of the set of nodes whose current degrees are at least $d_1$ (or, in case there
are fewer than $s$ such nodes, it contains all such nodes).
To this end, as soon as the degree of an $A$-vertex reaches $d_1$, the vertex is introduced into the reservoir 
with an appropriate probability (and another vertex is removed if the reservoir is already full), so as to maintain a uniform sample. 
Once a vertex is introduced into the reservoir, incident edges to this vertex are collected until $d_2$ such edges are found.

\begin{algorithm}
 \begin{algorithmic}[1]
   \REQUIRE Integral degree bounds $d_1$ and $d_2$, reservoir size $s$
   \STATE $R \gets \{ \}$ \COMMENT{reservoir}, $S \gets \{ \}$ \COMMENT{collected edges}, $x \gets 0$ \COMMENT{counter for nodes of degree $\ge d_1$}
   \WHILE{stream not empty}
    \STATE Let $ab$ be next edge in stream
    \STATE Increment degree $\deg(a)$ by one
    \IF[candidate to be inserted into reservoir]{$\deg(a) = d_1$}
     \STATE $x \gets x+1$
     \IF[reservoir not yet full]{$|R| < s$}
      \STATE $R \gets R \cup \{a\}$
     \ELSE[reservoir full]
      \IF[insert $a$ into reservoir with prob. $\frac{s}{x}$]{\textsc{Coin}$\displaystyle (\frac{s}{x})$}
      \STATE Let $a'$ be a uniform random element in $R$
      \STATE $R \gets (R \setminus \{a'\}) \cup \{a\}$, delete all edges incident to $a'$ from $S$
      \ENDIF
     \ENDIF     
    \ENDIF
    \IF[collect edge]{$a \in R$ and $\deg_S(a) < d_2$}
     \STATE $S \gets S \cup \{ab\}$
    \ENDIF    
   \ENDWHILE
   \RETURN Arbitrary neighbourhood among those of size $d_2$ in $S$, if there is none return $\texttt{fail}$
 \end{algorithmic}
 \caption{\textsc{Deg-Res-Sampling}$(d_1, d_2, s)$ \label{alg:deg-res-sampling}}
\end{algorithm}
The description of Algorithm~\ref{alg:deg-res-sampling} assumes that we have a function $\textsc{Coin}(p)$ to our disposal that 
outputs \texttt{true} with probability $p$ and \texttt{false} with probability $1-p$.

Disregarding the maintenance of the vertex degrees, the algorithm uses space $O(s d_2 \log n)$ since at most $d_2$ neighbours 
for each vertex in the reservoir are stored, and we account space $\Order(\log n)$ for storing an edge.

%\begin{lemma}
% Suppose that Algorithm~\textsc{Deg-Res-Sampling}$(d_1, d_2, s)$ succeeds. Then it reports a uniform random vertex from set $\{ a \in A \ : \ \deg(a) \ge d_1 + d_2 -1 \}$.
%\end{lemma}
%\begin{proof}
% Let's see whether and how this is needed later.
%\end{proof}

\begin{lemma}\label{lem:deg-res-samp}
 Suppose that $G$ contains at most $n_1$ $A$-nodes of degree at least $d_1$ and at least $n_2$ $A$-nodes of degree at least $d_1 + d_2 - 1$. Then,
 Algorithm~\textsc{Deg-Res-Sampling}$(d_1, d_2, s)$ succeeds with probability at least
 $$1 - (1 - \frac{s}{n_1})^{n_2} \ge 1 - e^{-\frac{s n_2}{n_1}} \ .$$ 
 %Furthermore, the algroithm uses space $O(s \cdot d_2)$ disregarding the space used for storing the nodes' degrees.
\end{lemma}
\begin{proof}
 Let $D \subseteq V$ be the set of vertices of degree at least $d_1$ (then $|D| \le n_1$). 
 First, suppose that $d_1 \le s$. Then the algorithm stores all nodes of degree at least $d_1$ (including all nodes of degree $d_1 + d_2-1$)
 and collects its incident edges (except the first $d_1-1$ such edges). Hence, a neighbourhood of size $d_2$ is necessarily found.
 
 Otherwise, by well-known properties of reservoir sampling (e.g. \cite{v85}), at the end of the algorithm the set $R$ 
 constitutes a uniform random sample of $D$ of size $s$. The probability that no node of degree at least 
 $d_1 + d_2-1$ is sampled is at most:
 
 \begin{eqnarray*}
  \frac{{n_1 - n_2 \choose s}}{{n_1 \choose s}} & = & \frac{(n_1 -n_2)! (n_1 - s)!}{(n_1 - n_2 - s)! n_1 !} \\ 
  & = & \frac{(n_1 - s) \cdot (n_1 - s - 1) \cdot \ldots \cdot (n_1 - s - n_2 + 1)}{n_1 \cdot (n_1 - 1) \cdot \ldots \cdot (n_1 - n_2 + 1)} \\
  & \le & \left( \frac{n_1 - s}{n_1} \right)^{n_2} = (1 - \frac{s}{n_1})^{n_2} \le e^{- \frac{s n_2}{n_1}} \ . 
 \end{eqnarray*} 
\end{proof}

\subsection{Main Algorithm}
Our main algorithm runs the subroutine presented in the previous subsection in parallel for multiple different threshold values $d_1$.
We will prove that the existence of a node of degree $d$ implies that at least one of these runs will succeed with high probability.

\begin{algorithm}
 \begin{algorithmic}
  \REQUIRE Degree bound $d$, approximation factor $\alpha$
  %\STATE $c \gets \lceil \frac{\log n}{\log \left( \frac{s}{\ln(n)d} \right)} \rceil$
  \STATE $s \gets \lceil \ln(n) \cdot n^{\frac{1}{\alpha}} \rceil$
%  \STATE Run \textsc{Deg-Res-Sampling}$(1, \frac{d}{c}, \frac{sc}{d})$
  \STATE \textbf{for} $i = 0 \dots \alpha-1$ \textbf{do in parallel} 
 \STATE $\quad$ $(a_i, S_i) \gets $ \textsc{Deg-Res-Sampling}$(\max \{1, i \cdot \frac{d}{\alpha} \}, \frac{d}{\alpha}, s)$
 \RETURN Any neighbourhood $(a_i, S_i)$ among the successful runs
\end{algorithmic}
\caption{$\alpha$-approximation Streaming Algorithm for \textsf{FEwW} \label{alg:one-pass}}
\end{algorithm}

\begin{theorem} \label{thm:insertion-only-upper-bound}
 Suppose that the input graph $G=(A, B, E)$ contains at least one $A$-node of degree at least $d$. For every integral $\alpha \ge 2$, 
 Algorithm~\ref{alg:one-pass} finds a neighbourhood of size $\frac{d}{\alpha}$ with probability at least $1-\frac{1}{n}$ and uses space
 $$\Order(n \log n + n^{\frac{1}{\alpha}} d \log^2 n) \ .$$
% $$c = \lceil \frac{\log n}{\log \left( \frac{s}{\ln(n)d} \right)} \rceil  \ .$$ 
 %Furthermore, Algorithm~\ref{alg:one-pass} uses space $O(n \log n + sc)$. 
\end{theorem}
\begin{proof}
Concerning the space bound, the algorithm needs to keep track of the degrees of all $A$-vertices which requires space $O(n \log n)$ (using
the assumption $m = \poly n$). The algorithm runs the subroutine \textsc{Deg-Res-Sampling} (Algorithm~\ref{alg:deg-res-sampling}) $\alpha$ times in parallel. Each of these runs
requires space $\Order(s \cdot \frac{d}{\alpha} \log n)$. Besides the vertex degrees, we thus have an additional space requirement of 
$O(s \cdot d \log n) = O(n^{\frac{1}{\alpha}} d \log^2 n)$ bits, which justifies the space requirements.

Concerning correctness, let $n_0$ be the number of $A$-nodes of degree at least $1$, and for $i \ge 1$, let $n_i$ be the number of $A$-nodes of degree at least
$i \cdot \frac{d}{\alpha}$. Observe that $n \ge n_0 \ge n_1 \ge n_2 \ge \dots \ge n_{\alpha} \ge 1$, where the last inequality follows from 
the assumption that the input graph contains at least one $A$-node of degree at least $d$.
 
 We will prove that at least one of the runs succeeds with probability at least $1 - \frac{1}{n}$. For the sake of a contradiction,
 assume that the error probability of every run is strictly larger than $\frac{1}{n}$. Then, using Lemma~\ref{lem:deg-res-samp}, we obtain for
 every $0 \le i \le \alpha-1$:
 \begin{eqnarray*}
  e^{-\frac{s n_{i+1}}{ n_i}} & > & \frac{1}{n} \ , \mbox{ which implies} \\
  n_{i+1} & < & \frac{\ln(n) n_i}{s} \ . 
%  \frac{sc n_{i+1}}{d n_i} & \le & \ln(n) \\  
 \end{eqnarray*} 
 Since $n_0 \le n$ we obtain:
 \begin{eqnarray*}
  n_i & < & n \left( \frac{\ln n }{s} \right)^i \ ,
 \end{eqnarray*}
 and since $n_c \ge 1$ we have:
$$1 <  n \left( \frac{\ln n}{s} \right)^{\alpha} \quad \mbox{which implies} \quad s < n^{\frac{1}{\alpha}} \ln n \ .$$
However, since the reservoir size in Algorithm~\ref{alg:one-pass} is chosen to be $\lceil n^{\frac{1}{\alpha}} \ln n \rceil$, we obtain 
a contradiction. Hence, at least one run succeeds with probability $1 - 1/n$.
\end{proof}

\subsection{Extension to \textsf{Star Detection}}
Streaming algorithms for \textsf{FEwW} can be used to solve \textsf{Star Detection} with a small increase in space and approximation ratio.

\begin{lemma} \label{lem:neighbourhood-detection-to-star-detection}
 Let $\mathbf{A}$ be a one-pass $\alpha$-approximation streaming algorithm for \textsf{FEwW} with space $s(n,d)$ that succeeds with probability $1-\delta$. Then there exists
 a one-pass $(1+\epsilon)\alpha$-approximation streaming algorithm for \textsf{Star Detection} with space $\Order(s(n,n) \cdot \log_{1+\epsilon} n)$ that succeeds with probability $1-\delta$.
\end{lemma}
\begin{proof}
 Let $G=(V, E)$ be the graph described by the input stream in an instance of \textsf{Star Detection}. 
We use $\Order(\log_{1+\epsilon} n)$ guesses $\Delta' \in \{1, 1+\epsilon, (1+\epsilon)^2,  \dots, (1+\epsilon)^{\lceil \log_{1+\epsilon} n \rceil} \}$ for $\Delta$, the maximum degree in the input graph. For each guess $\Delta'$ we run algorithm $\mathbf{A}$ for 
\textsf{FEwW} with threshold value $d= \Delta'$ on the bipartite graph $H=(V, V, E')$, where for every edge $uv$ in the input stream, we include the two edges $uv$ and $vu$ into $H$. 
 
 Consider the run with the largest value for $\Delta'$ that is not larger than $\Delta$. Then, $\Delta' \ge \Delta / (1+\epsilon)$.
 This run finds a neighbourhood of size at least $\Delta' / \alpha \ge \Delta / (\alpha (1+\epsilon))$ and thus a large star in $G$.
 %We thus obtain a $(1+\epsilon)c$-approximation algorithm with space $\OrderT(\log_{1+\epsilon}(n) n^{1+\frac{1}{c}})$.
 %Using any fixed constant for $\epsilon$ and $c=\log n$, this construction yields a $(1+\epsilon) \log n$-approximation 
 %semi-streaming algorithm for approximating the largest star in a general graph.
\end{proof}

The previous result in combination with Theorem~\ref{thm:insertion-only-upper-bound} can be used to obtain a semi-streaming algorithm for \textsf{Star Detection} (by using any fixed constant $\epsilon$ and $\alpha = \log n$ in the previous lemma).
\begin{corollary} \label{cor:star-detection}
 There is a semi-streaming $\Order(\log n)$-approximation algorithm for \textsf{Star Detection} that succeeds with high probability.
\end{corollary}
% \begin{proof} 
% Let $G=(V, E)$ be the graph described by the input stream in an instance of \textsf{Star Detection}. 
% We use $\Order(\log_{1+\epsilon} n)$ guesses $\Delta' \in \{1, 1+\epsilon, (1+\epsilon)^2,  \dots, (1+\epsilon)^{\lceil \log_{1+\epsilon} n \rceil} \}$ for $\Delta$, 
%  the maximum degree in the input graph. For each guess $\Delta'$ we run our insertion-only algorithm for \textsf{Neighborhood Detection} 
%  with threshold value $d= \Delta'$ on the bipartite graph $H=(V, V, E')$, where for every edge $uv$ in the input stream, 
%  we include the two edges $uv$ and $vu$ into $H$. 
%  
%  Consider the run with the largest value for $\Delta'$ that is not larger than $\Delta$. Then, $\Delta' \ge \Delta / (1+\epsilon)$.
%  This run detects a neighbourhood of size at least $\Delta' / c \ge \Delta / (c (1+\epsilon))$ and thus a star of this size in $G$.
%  We thus obtain a $(1+\epsilon)c$-approximation algorithm with space $\OrderT(\log_{1+\epsilon}(n) n^{1+\frac{1}{c}})$.
%  Using any fixed constant for $\epsilon$ and $c=\log n$, this construction yields a $(1+\epsilon) \log n$-approximation 
%  semi-streaming algorithm for approximating the largest star in a general graph.
% \end{proof}

\section{Lower Bound for Insertion-only Streams} \label{sec:lb-insertion-only}
In this section, we first point out that a simple $\Omega(n / \alpha^2)$ lower 
bound follows from the one-way communication complexity of a multi-party version of the \textsf{Set-Disjointness} problem. 
Next, we give some important inequalities involving entropy and mutual information that are used subsequently.
Then, we prove our main lower bound result of this section. 
To this end, we first define the multi-party one-way communication problem \textsf{Bit-Vector Learning} 
and prove a lower bound on its communication complexity. We then show that a streaming algorithm for \textsf{FEwW} 
yields a protocol for \textsf{Bit-Vector Learning}, which gives the desired lower bound.

\subsection{An $\Omega(n/\alpha^2)$ Lower Bound via Multi-party Set-Disjointness}
Consider the one-way multi-party version of the well-known \textsf{Set-Disjointness} problem:
\begin{problem}[$\textsf{Set-Disjointness}_p$]
 $\textsf{Set-Disjointness}_p$ is a $p$-party communication problem where every party $i$ holds 
 a subset $S_i \subseteq \mathcal{U}$ of a universe $\mathcal{U}$ of size $n$. The parties are given 
 the promise that either their sets are pairwise disjoint, i.e., $S_i \cap S_j = \varnothing$ for every 
 $i \neq j$, or they {\em uniquely} intersect, i.e., $|\cap_i S_i| = 1$. The goal is to determine which is the case.
\end{problem}
It is known that every $\epsilon$-error protocol for $\textsf{Set-Disjointness}_p$ requires a total communication of 
$\Omega(n / p)$ bits \cite{cks03}. Since our notion of one-way multi-party communication complexity measures the 
maximum length of any message sent in an optimal protocol, we obtain:
$$R_{\epsilon}^{\rightarrow}(\textsf{Set-Disjointness}_p) = \Omega(n / p^2) \ .$$

We now argue that an algorithm for $\textsf{FEwW}$ can be used to solve $\textsf{Set-Disjointness}_p$.

\begin{theorem} \label{thm:lb-set-disjointness}
 Every $\alpha/1.01$-approximation streaming algorithm for \textsf{FEwW}$(n, d)$ requires space $\Omega(n / \alpha^2)$, for any integral $\alpha$ and for any $d = k \cdot \alpha$, where $k$ is a positive integer.
\end{theorem}
\begin{proof}
 Let $(S_1, S_2, \dots, S_p)$ be an instance of $\textsf{Set-Disjointness}_p$. For $\alpha = p / 1.01$, let $\mathbf{A}$ be an $\alpha$-approximation
 streaming algorithm for \textsf{FEwW}, and let $d = k \cdot p$, for some integer $k \ge 1$. The parties use 
 $\mathbf{A}$ to solve $\textsf{Set-Disjointness}_p$ as follows: The $p$-parties construct a graph $G=(\mathcal{U}, B, E)$ 
 with $B = [d]$ and $E = \cup_{i=1}^p E_i$. 
 Each party $i$ translates $S_i$ into the set of edges $E_i$ where for each $u \in S_i$ the edges 
 $\{ ub \ : \ b \in \{ (i-1) d/p + 1, \dots, i d/p \} \}$
 are included in $E_i$. Observe that $\Delta = d/p = k$ if all sets $S_i$ are pairwise disjoint, and $\Delta = d = k \cdot p$ 
 if they uniquely intersect. Party $1$ now simulates $\mathbf{A}$ on their edges $E_1$, sends the resulting memory state to party $2$
 who continues running $\mathbf{A}$ on $E_2$. This continues until party $p$ completes the algorithm. Since $\mathbf{A}$
 is a $p/1.01$-approximation algorithm, if the sets uniquely intersect, the output of the algorithm is a neighbourhood
 of size at least $\lceil \frac{\Delta}{\alpha} \rceil = \lceil 1.01 \cdot k \rceil \ge k+1$. If the sets are disjoint, then no neighbourhood is of size larger than $k$. The 
 last party can thus distinguish between the two cases and solve $\textsf{Set-Disjointness}_p$. Since at least one message 
 used in the protocol is of length $\Omega(n / p^2)$, $\mathbf{A}$ uses space $\Omega(n / p^2) = \Omega(n / \alpha^2)$.
 
\end{proof}

\subsection{Inequalities Involving Entropy and Mutual Information} \label{sec:prelim}

%Let $A$ be a random variable
%distributed according to $\mathcal{D}$. The {\em Shannon Entropy} of $A$ is denoted by $H_{\mathcal{D}}(A)$,
%or simply $H(A)$ if the distribution $\mathcal{D}$ is clear from the context. The mutual information
%of two jointly distributed random variables $A, B$ with distribution $\mathcal{D}$ is denoted by 
%$I_{\mathcal{D}}(A, B) := H_{\mathcal{D}}(A) - H_{\mathcal{D}}(A \ | \ B)$ 
%(again, the subscript may be dropped if the distribution is clear from the context),
%where $H_{\mathcal{D}}(A \ | \ B)$ is the entropy of $A$ conditioned on $B$.

In the following, we will use various properties of entropy and mutual information. The most important
ones are listed below: (let $A,B,C$ be jointly distributed random variables)
\begin{enumerate}[leftmargin=0.6cm]
 \item \emph{Chain Rule for Entropy:} $H(A B \ | \ C) = H(A \ | \ C) + H(B \ | \ AC)$
 \item \emph{Conditioning reduces Entropy:} $H(A) \ge H(A \ | \ B) \ge H(A \ | \ BC)$  \label{rule:4}
 %\item \emph{Conditional Entropy:} $H(A \ | \ B) = H(A)$ iff $A$ and $B$ are independent. \label{rule:2}
 \item \emph{Chain Rule for Mutual Information:} $I(A \ : \ B C) = I(A \ : \ B) + I(A \ : \ C \ | \ B)$ 
 \item \emph{Data Processing Inequality:}\footnote{Technically the data processing inequality is more general, however, the inequality stated here is sufficient for our purposes.} Suppose that $C$ is a deterministic function of $B$. Then: 
 $I(A \ : \ B) \ge I(A \ : \ C)$
 \item \emph{Independent Events:} Let $E$ be an event independent of $A, B, C$. Then: $I(A \ : \ B \ | \ C, E) = I(A \ : \ B \ | \ C)$
\end{enumerate}
We will also use the following claim: (see Claim 2.3. in \cite{akl16} for a proof)
\begin{lemma} \label{lem:indep-variables}
 Let $A,B,C,D$ be jointly distributed random variables so that $A$ and $D$ are independent conditioned on $C$. Then:
 $I(A \ : \ B \ | \  C D) \ge I(A \ : \ B \ | \ C)$.
\end{lemma}

\subsection{Hard Communication Problem: \textsf{Bit-Vector Learning}}

We consider the following one-way $p$-party communication game:

\begin{problem}[\textsf{Bit-Vector Learning}$(p, n, k)$]\label{prob:bit-vector-learning}
 Let $X_1 = [n]$ and for every $2 \le i \le p$, let $X_i$ be a uniform random subset of 
 $X_{i-1}$ of size $n_i = n^{1-\frac{i-1}{p-1}}$.
 Furthermore, for every $1 \le i \le p$ and every $1 \le j \le n$, let 
 $Y_i^j \in \{0, 1 \}^k$ be a uniform random bit-string if $j \in X_i$, and let $Y_i^j = \epsilon$ 
 (the empty string) if $j \notin X_i$.
 %Let $P^j$ be the largest integer $q$ such that $j \in X_q$. 
 For $j \in [n]$, let $Z^j = Y_1^j \circ Y_2^j \circ \dots \circ Y_p^j$ be the bit string obtained by concatenation.
  
 Party $i$ holds $X_i$ and $Y_i := Y_i^1, \dots, Y_i^n$. Communication is one way from party $1$ through party $p$ and 
 party $p$ needs to output an index $I \in [n]$ and at least $1.01k$ bits from  
 string $Z^I$. \footnote{More formally, the output is an index $I \in [n]$ and a set of tuples 
 $\{ (i_1, \tilde{Z}_{1}), (i_2, \tilde{Z}_{2}), \dots \}$
 of size at least $1.01k$ with $i_j \neq i_k$ for every $j \neq k$ so that $Z^I[i_j] = \tilde{Z}_{j}$, for every $j$. } %We denote the output by $Out$.
\end{problem}
Observe that the previous definition also defines an input distribution. All subsequent entropy and mutual information 
terms refer to this distribution. 
An example instance of \textsf{Bit-Vector Learning}$(3, 4, 5)$ is given in Figure~\ref{fig:bit-vector-learning-main}.

\begin{figure}[h!]
\begin{tikzpicture}[scale=0.6]
  \node (a) at (0,0) {\begin{minipage}{2.7cm} \begin{center}\textbf{Alice} \\  \vspace{0.2cm} $X_1 = \{1, 2, 3, 4 \}$ \end{center} 
  \begin{eqnarray*}
   Y_1^1 & = & 10010 \\
   Y_1^2 & = & 01000 \\
   Y_1^3 & = & 01011 \\
   Y_1^4 & = & 01111 
  \end{eqnarray*}
\end{minipage}};

\node (a) at (5,0) {\begin{minipage}{2.7cm} \begin{center}\textbf{Bob} \\  \vspace{0.2cm} $X_2 = \{1, 4 \}$ \end{center} 
  \begin{eqnarray*}
   Y_2^1 & = & 11011 \\
   Y_2^2 & = & \epsilon \\
   Y_2^3 & = & \epsilon \\
   Y_2^4 & = & 01010 
  \end{eqnarray*}
\end{minipage}};

\node (a) at (10,0) {\begin{minipage}{2.7cm} \begin{center}\textbf{Charlie} \\  \vspace{0.2cm} $X_3 = \{ 4 \}$ \end{center} 
  \begin{eqnarray*}
   Y_3^1 & = & \epsilon \\
   Y_3^2 & = & \epsilon \\
   Y_3^3 & = & \epsilon \\
   Y_3^4 & = & 00011 
  \end{eqnarray*}
\end{minipage}};

\draw[->] (1.8,1.7) -- (3.1,1.7) node[midway,sloped, above] {$M_1$};
\draw[->] (6.8,1.7) -- (8.1,1.7) node[midway,sloped, above] {$M_2$};
 \end{tikzpicture}
\caption{Example instance of \textsf{Bit-Vector Learning}$(3, 4, 5)$. Charlie needs to output at least $1.01 \cdot 5$ positions (i.e., at least $6$ positions) of one 
of the strings $Z^1 = 1001011011$, $Z^2 = 01000$, $Z^3 = 01011$, or $Z^4 = 011110101000011$. 
%Observe that $P^1 = 2$, $P^2 = 1$, $P^3 = 1$, and $P^4 = 3$. 
\label{fig:bit-vector-learning-main}}
\end{figure}

In the following, for a subset $S \subseteq [n]$, we will use the notation $Y_i^{S}$, which refers to the strings $Y_i^{s_1}, Y_i^{s_2}, \dots, Y_i^ {s_{|S|}}$, where $S = \{s_1, s_2, \dots, s_{|S|} \}$.

%In the following, we will prove the following theorem:

%\begin{theorem}
% Every randomized constant error protocol for \textsf{Bit-Vector Learning}$(p, n, k)$ requires at least one message of size $XYZ$.
%\end{theorem}

Observe further that there is a protocol that requires no communication and outputs an index $I$ and $k$ bits of $Z^I$: Party $p$ simply outputs
the single element $I \in X_p$ together with the bit string $Y_p^I$. As our main result of this section we show that every protocol
that outputs at least $1.01k$ bits of any string $Z^i$ ($i \in [n]$) needs to send at 
least one message of length $\Omega( \frac{k n^{\frac{1}{p-1}}}{p})$.

\textit{Remark:} For ease of presentation, we will only consider values of $n$ so that $n^{\frac{1}{p-1}}$ is integral. This condition implies that $n_{i+1} \mid n_i$ for every $1 \le i \le p-1$ since $\frac{n_i}{n_{i+1}} = n^{\frac{1}{p-1}}$. The reason for this restriction is that we will apply Baranyai's theorem \cite{b79}, which is stated as Theorem~\ref{thm:baranyai} below and requires this property. This can be avoided using additional case distinctions.

%Before proving this lower bound, we will first show that a streaming algorithm for \textsf{Neighborhood Detection} 
%can be used to solve \textsf{Bit-Vector Learning}.

\subsection{Lower Bound Proof for \textsf{Bit-Vector Learning}} \label{sec:lb-bit-vector-learning}
Fix now an arbitrary deterministic protocol $\Pi$ for \textsf{Bit-Vector Learning}$(p, n, k)$ with distributional error $\epsilon$.
Let $Out = (I, \tilde{Z}^I)$ denote the neighbourhood outputted by the protocol. Furthermore, denote by $M_i$ the message
sent from party $i$ to party $i+1$. Throughout this section let $s = \max_i |M_i|$.

Since the last party correctly identifies $1.01k$ bits of $Z^I$, the mutual information between $Z^I$ and all random 
variables known to the last party, that is, $M_{p-1}, X_p$ and $Y_p$, needs to be large. This is proved in the next lemma:

\begin{lemma} \label{lem:fano-type}
 We have:
 $$I(M_{p-1} X_p Y_p \ : \  Z^I) \ge (1-\epsilon) 1.01k - 1 \ . $$
\end{lemma}
\begin{proof}
We will first bound the term $I(Out \ : \ Z^I) = H(Z^I) - H(Z^I \ | \ Out)$. To this end, let 
$E$ be the indicator variable of the event that the protocol errs. Then, $\Pr [E=1] \le \epsilon$. We have:
 \begin{eqnarray}
  \nonumber H(E, Z^I \ | \ Out) & = & H(Z^I \ | \ Out) + H(E \ | \ Out, Z^I ) \\
  & = & H(Z^I \ | \ Out) \ , \label{eqn:209}
 \end{eqnarray}
 where we used the chain rule for entropy and the observation that $E$ is fully determined
 by $Out$ and $Z^I$ which implies $H(E \ | \ Out, Z^I ) = 0$. Furthermore,
 \begin{eqnarray} 
  \nonumber H(E, Z^I \ | \ Out) & = & H(E \ | \ Out) + H(Z^I  \ | \ E, Out) \\
  & \le & 1 +  H(Z^I \ | \ E, Out) \ , \label{eqn:210}
 \end{eqnarray}
using the chain rule for entropy and the bound $H(E \ | \ Out) \le H(E) \le 1$ (conditioning reduces entropy). 
From Inequalities~\ref{eqn:209} and \ref{eqn:210}
we obtain:
\begin{eqnarray} \label{eqn:230}
 H(Z^I \ | \ Out) & \le & 1 +  H(Z^I \ | \ E, Out) \ .
\end{eqnarray}
Next, we bound the term $H(Z^I \ | \ E, Out)$ as follows:
\begin{eqnarray} 
 \nonumber H(Z^I \ | \ E, Out) & = &  \Pr \left[ E = 0 \right] H(Z^I \ | \ Out, E = 0) \\
 & & + \  \Pr \left[ E = 1 \right] H(Z^I \ | \ Out, E = 1)  \ . \label{eqn:211}
\end{eqnarray}
Concerning the term $H(Z^I \ | \ Out, E = 0)$, %observe that $Z^I$ is a bit string of length at most 
%$p \cdot k$ and is uniformly distributed. 
since no error occurs, $Out$ already determines at least $1.01k$ bits of $Z^I$.
We thus have that $H(Z^I \ | \ Out, E = 0) \le H(Z^I) - 1.01 k$.
We bound the term $H(Z^I \ | \ Out, E = 1)$ by $H(Z^I \ | \ Out, E = 1) \le H(Z^I)$ (since conditioning can only decrease entropy).
The quantity $H(Z^I \ | \ E, Out)$ can thus be bounded as follows:
\begin{eqnarray} 
 \nonumber H(Z^I \ | \ E, Out) & \le & (1-\epsilon) (H(Z^I) - 1.01k) + \epsilon H(Z^I) \\
 & = & H(Z^I) - (1-\epsilon) 1.01k \ . \label{eqn:212}
\end{eqnarray}
Next, using Inequalities~\ref{eqn:230} and \ref{eqn:212}, we thus obtain:
\begin{eqnarray*}
 I(Out \ : \ Z^I) & = & H(Z^I) - H(Z^I \ | \ Out) \\
 & \ge & H(Z^I) - 1 - H(Z^I \ | \ E, Out) \\
 & \ge & H(Z^I) - 1 - ( H(Z^I) - (1-\epsilon) 1.01k ) \\
 & = & (1-\epsilon) 1.01k - 1\ .
\end{eqnarray*}
Last, observe that $Out$ is a function of $M_{p-1}, X_p$, and $Y_p$. The result then follows from the data processing inequality.
\end{proof}

%In the following, we will denote by $Y_i^{X_{i+1}}$ the set of random variables $Y_i^J$ with $J \in X_{i+1}$. 
%In Lemma~\ref{lem:lb-bit-vector-learning} we will prove the following statement:
%$$s \le S \Rightarrow I(M_{p-1} X_p Y_p \ : \ Y_1^{X^2} Y_2^{X^3} \dots Y_{p-1}^{X_p}) \le 0.001k \ .$$
%However, if this specific mutual information is bounded as in the previous inequality, then we obtain a contradiction to 
%Lemma~\ref{lem:fano-type}, since:

Next, since the set $X_i$ is a uniform random subset of $X_{i-1}$, we will argue in Lemma~\ref{lem:bit-vector-learning-key-lemma} 
that the message $M_{i-1}$ can only contain a limited amount of information about the bits $Y_{i-1}^{X_i}$. 
This will be stated as a suitable conditional mutual information expression that will be used later. The proof of 
Lemma~\ref{lem:bit-vector-learning-key-lemma} relies on Baranyai's theorem \cite{b79}, which in its original form states
that every complete regular hypergraph is $1$-factorisable, i.e., the set of hyperedges can be partitioned into 
$1$-factors. We restate this theorem as Theorem~\ref{thm:baranyai} in a form that is more suitable for our purposes.

\begin{theorem}[Baranyai's theorem \cite{b79} - rephrased] \label{thm:baranyai}
 Let $k,n$ be integers so that $k \mid n$. Let $S \subseteq 2^{[n]}$ be the set consisting of all subsets of $[n]$ of cardinality $k$. 
 Then there exists a partition of $S$ into $|S| \frac{k}{n}$ subsets $S_1, S_2, \dots, S_{|S| \frac{k}{n}}$ such that:
 \begin{enumerate}
  \item $|S_i| = \frac{n}{k}$, for every $i$,
  \item $S_i \cap S_j = \varnothing$, for every $i \neq j$, and
  \item $\bigcup_{x \in S_i} x = [n]$, for every $i$.
 \end{enumerate}
\end{theorem}

\begin{lemma} \label{lem:bit-vector-learning-key-lemma}
 Suppose that $n_i \mid n_{i-1}$. Then:
 $$I(M_{i-1} \ : \ Y_{i-1}^{X_i} \ | \ X_i) \le \frac{n_i}{n_{i-1}}|M_{i-1}|  \ .$$ 
\end{lemma}
\begin{proof}
First, using Lemma~\ref{lem:indep-variables}, we obtain $I(M_{i-1} \ : \ Y_{i-1}^{X_i} \ | \ X_i) \le I(M_{i-1} \ : \ Y_{i-1}^{X_i} \ | \ X_i X_{i-1})$ (observe that $Y_{i-1}^{X_i}$ and $X_{i-1}$ are independent conditioned on $X_{i}$).
Then, using the definition of conditional mutual information, we rewrite as follows:
\begin{align}
 \nonumber I(M_{i-1}  & : Y_{i-1}^{X_i} \ | \ X_i X_{i-1}) \\
 = &  \Exp_{x_{i-1} \gets X_{i-1}} \Exp_{x_i \gets X_i} I(M_{i-1} \ : \ Y_{i-1}^{X_i} \ | \ X_i = x_i, X_{i-1} = x_{i-1}) \nonumber \\
 = & \Exp_{x_{i-1} \gets X_{i-1}} \Exp_{x_i \gets X_i} I(M_{i-1} \ : \ Y_{i-1}^{x_i} \ | \ X_{i-1} = x_{i-1}) \ . \label{eqn:958}
\end{align}
 Let $\mathcal{X}(x_{i-1})$ be the set of all subsets of $x_{i-1}$ of size $n_i$. %Then $|\mathcal{X}(x_{i-1})| = {n_{i-1} \choose n_i}$. 
 Observe that the distribution of $X_i$ is uniform among the elements $\mathcal{X}(x_{i-1})$.
 Next, since $n_i \mid n_{i-1}$, by Baranyai's theorem \cite{b79} as stated in Theorem~\ref{thm:baranyai}, the set $\mathcal{X}(x_{i-1})$ can be partitioned into 
 $|\mathcal{X}(x_{i-1})|\frac{n_{i}}{n_{i-1}}$ subsets $\mathcal{X}_1(x_{i-1}), \mathcal{X}_2(x_{i-1}), \dots$
 such that $\cup_{x \in \mathcal{X}_{j}(x_{i-1})} x = x_{i-1}$. Denote the elements of set $\mathcal{X}_j(x_{i-1})$ by 
 $x_j^1, x_j^2, \dots, x_j^{\frac{n_{i-1}}{n_i}}$. We thus have:
\begin{align}
\nonumber &  \Exp_{x_i \gets X_i} I(M_{i-1} \ : \ Y_{i-1}^{x_i} \ | \ X_{i-1} = x_{i-1}) \\ 
\nonumber & = \frac{1}{|\mathcal{X}(x_{i-1})|} \sum_{x_i \in \mathcal{X}(x_{i-1})} I(M_{i-1} \ : \ Y_{i-1}^{x_i} \ | \ X_{i-1} = x_{i-1}) \\
%\nonumber & & \quad \quad \quad \quad = \frac{1}{|\mathcal{X}(x_{i-1})|} \sum_{j \in [|\mathcal{X}(x_{i-1})|  \frac{n_{i}}{n_{i-1}}]} \sum_{x_i \in \mathcal{X}_j(x_{i-1})} I(M_{i-1} \ : \ Y_{i-1}^{x_i} \ | \ X_{i-1} = x_{i-1}) \\
\nonumber & = \frac{1}{|\mathcal{X}(x_{i-1})|} \sum_{j \in [|\mathcal{X}(x_{i-1})| \frac{ n_{i}}{n_{i-1}}]} \sum_{\ell \in [\frac{n_{i-1}}{n_i}]} I(M_{i-1} \ : \ Y_{i-1}^{x_j^{\ell}} \ | \ X_{i-1} = x_{i-1}) \\
\nonumber & \le \frac{1}{|\mathcal{X}(x_{i-1})|} \sum_{j \in [|\mathcal{X}(x_{i-1})|\frac{n_{i}}{n_{i-1}}]} \sum_{\ell \in [\frac{n_{i-1}}{n_i}]} I(M_{i-1} \ : \ Y_{i-1}^{x_j^{\ell}} \ | \ Y_{i-1}^{x_j^{1}} \dots \\
& \quad \quad \quad \quad \quad \quad \quad \quad \quad \quad \quad \quad \quad \quad \quad \quad \quad  \dots Y_{i-1}^{x_j^{\ell - 1}}, X_{i-1} = x_{i-1}) \nonumber \\
\nonumber & = \frac{1}{|\mathcal{X}(x_{i-1})|} \sum_{j \in [|\mathcal{X}(x_{i-1})| \frac{n_{i}}{n_{i-1}}]} I(M_{i-1} \ : \ Y_{i-1} \ | \  X_{i-1} = x_{i-1}) \\
 & = \frac{n_{i}}{n_{i-1}} I(M_{i-1} \ : \ Y_{i-1} \ | \  X_{i-1} = x_{i-1})  \ , \label{eqn:376} 
\end{align}
where we used Lemma~\ref{lem:indep-variables} to obtain the first inequality and the chain rule for mutual information for 
the subsequent equality.
Combining Inequalities~\ref{eqn:958} and \ref{eqn:376}, we obtain:
\begin{align*}
 & I(M_{i-1} \ : \ Y_{i-1}^{X_i} \ | \ X_i X_{i-1}) \\
 & \le  \Exp_{x_{i-1} \gets X_{i-1}} \frac{n_{i}}{n_{i-1}} I(M_{i-1} \ : \ Y_{i-1} \ | \  X_{i-1} = x_{i-1}) \\
 & = \frac{n_{i}}{n_{i-1}} I(M_{i-1} \ : \ Y_{i-1} \ | \  X_{i-1}) \le \frac{n_{i}}{n_{i-1}} H(M_{i-1}) \le \frac{n_{i}}{n_{i-1}} |M_{i-1}| \ .
\end{align*}

\end{proof}

The next lemma shows that the last party's knowledge about the crucial bits $Y_1^{X_2}, Y_2^{X_3}, \dots, Y_{p-1}^{X_p}$ is
limited.

\begin{lemma}\label{lem:lb-bit-vector-learning}
The following holds: (recall that $s = \max_{i} |M_i|$)
 $$I(M_{p-1} X_p Y_p \ : \ Y_1^{X_2} Y_2^{X_3} \dots Y_{p-1}^{X_p}) \le \frac{s(p-1)}{ n^{\frac{1}{p-1}}} \ .$$
\end{lemma}
\begin{proof}
Let $3 \le i \le p$ be an integer. Then:
 \begin{align}
  \nonumber & I(M_{i-1} X_i Y_i \ : \ Y_1^{X_2} \dots Y_{i-1}^{X_i}) \\
  & = I(X_i Y_i \ : \ Y_1^{X_2} \dots Y_{i-1}^{X_i}) + 
  I(M_{i-1}  \ : \ Y_1^{X_2} \dots Y_{i-1}^{X_i} \ | \ X_i Y_i) \nonumber \\
  & = 0 + I(M_{i-1}  \ : \ Y_1^{X_2} \dots Y_{i-1}^{X_i} \ | \ X_i) \ , \label{eqn:918}
 \end{align}
 where we first applied the chain rule, then used that the respective random variables are independent, and finally eliminated
 the conditioning on $Y_i$, which can be done since all other variables are independent with $Y_i$ (see Rule $5$ in Section~\ref{sec:prelim}).
Next, we apply the chain rule again, invoke Lemma~\ref{lem:bit-vector-learning-key-lemma}, and remove variables 
from the conditioning as they are independent with all other variables:
 \begin{align*}
  & I(M_{i-1}  \ : \ Y_1^{X_2} \dots Y_{i-1}^{X_i} \ | \ X_i) \\
  & = I(M_{i-1}  \ : \ Y_{i-1}^{X_i} \ | \ X_i) + I(M_{i-1}  \ : \ Y_1^{X_2} \dots Y_{i-2}^{X_{i-1}} \ | \ X_i Y_{i-1}^{X_i}) \\
  & \le |M_{i-1}| \frac{n_p}{n_{p-1}} + I(M_{i-1}  \ : \ Y_1^{X_2} \dots Y_{i-2}^{X_{i-1}} \ | \ Y_{i-1}^{X_i}) \ .
 \end{align*}  
Next, we bound the term $I(M_{i-1}  \ : \ Y_1^{X_2} \dots Y_{i-2}^{X_{i-1}} \ | \ Y_{i-1}^{X_i})$ by using the data processing inequality,
the chain rule, and remove an independent variable from the conditioning:
\begin{align*}
 & I(M_{i-1}  \ : \ Y_1^{X_2} \dots Y_{i-2}^{X_{i-1}} \ | \ Y_{i-1}^{X_i}) \\
 & \quad \quad \le I(M_{i-2} X_{i-1} Y_{i-1}  \ : \ Y_1^{X_2} \dots Y_{i-2}^{X_{i-1}} \ | \ Y_{i-1}^{X_i}) \\
 & \quad \quad = I(X_{i-1} Y_{i-1} \ : \ Y_1^{X_2} \dots Y_{i-2}^{X_{i-1}} \ | \ Y_{i-1}^{X_i}) \\
 & \quad \quad \quad \quad + \ I(M_{i-2}  \ : \ Y_1^{X_2} \dots Y_{i-2}^{X_{i-1}} \ | \ X_{i-1} Y_{i-1} Y_{i-1}^{X_i}) \\
 & \quad \quad  = 0 + I(M_{i-2}  \ : \ Y_1^{X_2} \dots Y_{i-2}^{X_{i-1}} \ | \ X_{i-1}) \ .
\end{align*}
We have thus shown:
 \begin{align} \nonumber
& I(M_{i-1}  \ : \ Y_1^{X_2} \dots Y_{i-1}^{X_i} \ | \ X_i) \\
\label{eqn:709}& \le |M_{i-1}| \frac{n_i}{n_{i-1}} + I(M_{i-2}  \ : \ Y_1^{X_2} \dots Y_{i-2}^{X_{i-1}} \ | \ X_{i-1}) \ .
\end{align} 
Using a simpler version of the same reasoning, we can show that:
 \begin{eqnarray} \label{eqn:710}
I(M_1  \ : \ Y_1^{X_2}  \ | \ X_2) \le |M_1| \frac{n_2}{n_1}  \ .  
 \end{eqnarray}
Using Equality~\ref{eqn:918} and Inequalities~\ref{eqn:709} and \ref{eqn:710}, we obtain:
 \begin{align*}
  & I(M_{p-1} \ : \ Y_1^{X_2} \dots Y_{p-1}^{X_p} X_p Y_p ) \\
  & \quad \quad \le s \left( \frac{n_p}{n_{p-1}} + \frac{n_{p-1}}{n_{p-2}} + \dots + \frac{n_2}{n_1}\right)  = \frac{(p-1) s}{ n^{\frac{1}{p-1}}} \ .
 \end{align*}
% recalling the definition $s = \max_i \{|M_i| \}$.
\end{proof}

Finally we are ready to prove the main result of this section.

\begin{theorem} \label{thm:bit-vector-learning}
For every $\epsilon < 0.005$, the randomized one-way communication complexity of \textsf{Bit-Vector Learning}$(p,n,k)$ is bounded as follows:
 \begin{align*}
R_{\epsilon}^{\rightarrow}(\textsf{Bit-Vector Learning}(p, n, k)) & \ge \frac{(0.005 k - 1)n^{\frac{1}{p-1}}}{p-1} \\
 & = \Omega( \frac{k n^{\frac{1}{p-1}}}{p}) \ .
 \end{align*}
 
 %Let $\mathbf{P}$ be an $\epsilon$-error protocol for \textsf{Bit-Vector Learning}$(p,n,k)$ that outputs at least $1.01k$ positions
 %from a string $Z^I$, for some $\epsilon < 0.01$. Then, at least one message sent in $\mathbf{P}$ is of length at least:
 %$$ \frac{(0.005 k - 1)n^{\frac{1}{p-1}}}{p-1} = \Omega( \frac{k n^{\frac{1}{p-1}}}{p}) \ .$$
\end{theorem}
\begin{proof}
 Let $q$ be the largest integer $i$ such that $Y_i^I \neq \epsilon$. Recall that by Lemma~\ref{lem:fano-type} 
 we have $I(M_{p-1} X_p Y_p \ : \  Z^I) \ge (1-\epsilon) 1.01k - 1$. However, we also obtain: 
 \begin{align*}
 I(M_{p-1} X_p Y_p \ : \  Z^I) & = I(M_{p-1} X_p Y_p \ : \  Y_1^I Y_2^I \dots Y_q^I) \\
  & = I(M_{p-1} X_p Y_p \ : \ Y_1^I Y_2^I \dots Y_{q - 1}^I) \\
  & \quad \quad + I(M_{p-1} X_p Y_p \ : \ Y_{q}^I \ | \ Y_1^I Y_2^I \dots Y_{q - 1}^I) \\
  & \le I(M_{p-1} X_p Y_p \ : \ Y_1^I Y_2^I \dots Y_{q - 1}^I) + H(Y_{q}^I) \\
  & \le I(M_{p-1} X_p Y_p \ : \ Y_1^{X_2} Y_2^{X_3} \dots Y_{p-1}^{X_p}) + k \\
  & \le \frac{(p-1) s}{ n^{\frac{1}{p-1}}} + k \ ,
\end{align*}
where we first applied the chain rule for mutual information, then observed that the variables
$Y_1^I Y_2^I \dots Y_{q - 1}^I$ are contained in the variables $Y_1^{X^2} Y_2^{X^3} \dots Y_{p-1}^{X_p}$, and then 
invoked Lemma~\ref{lem:lb-bit-vector-learning}. This is thus only possible if:
\begin{eqnarray*}
 (1-\epsilon) 1.01k - 1 \le \frac{(p-1) s}{ n^{\frac{1}{p-1}}} + k \ , 
\end{eqnarray*}
which, using $\epsilon < 0.005$, implies
\begin{eqnarray*}
 \frac{(0.005 k - 1)n^{\frac{1}{p-1}}}{p-1} \le s \ .
\end{eqnarray*}
Since we considered an arbitrary protocol $\Pi$, the result follows.
\end{proof}

\subsection{Reduction: \textsf{FEwW} to \textsf{Bit-Vector Learning}}
In this subsection, we show that a streaming algorithm for \textsf{FEwW} can be used to obtain a communication 
protocol for \textsf{Bit-Vector Learning}. The lower bound on the communication complexity of \textsf{Bit-Vector Learning} 
thus yields a lower bound on the space requirements of any algorithm for \textsf{FEwW}.

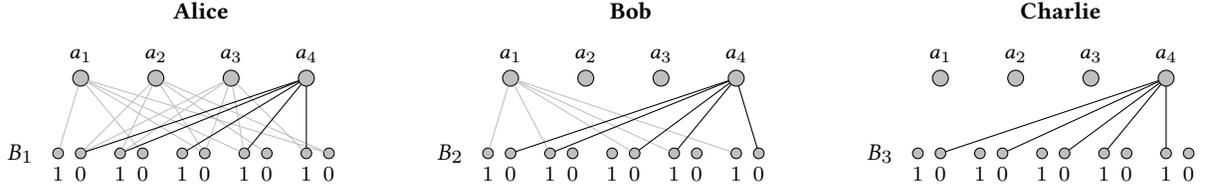
\begin{figure*}[t]
\begin{center}
  \begin{tikzpicture}
  \tikzstyle{vertex}=[circle,draw=black,fill=lightgray,minimum size=6pt,inner sep=0pt, outer sep=0pt]
  \tikzstyle{vertex2}=[circle,draw=black,fill=lightgray,minimum size=4pt,inner sep=0pt, outer sep=0pt]
  \tikzstyle{edge}=[draw=black]
  \tikzstyle{edge2}=[draw=lightgray]
 
  \node at (1.6, 0.9) {\textbf{Alice}};
  \node at (-0.8, -1) {$B_1$};
  \node [vertex] (v1) [label={\small $a_1$}] at (0,0) {};
  \node [vertex] (v2) [label={\small $a_2$}] at (1,0) {};  
  \node [vertex] (v3) [label={\small $a_3$}] at (2,0) {};  
  \node [vertex] (v4) [label={\small $a_4$}] at (3,0) {};   
  \node [vertex2] (x1) [label=below:{\small $1$}] at (-0.3,-1) {};
  \node [vertex2] (x2) [label=below:{\small $0$}] at (-0,-1) {};   
  \node [vertex2] (x3) [label=below:{\small $1$}] at (0.525,-1) {};
  \node [vertex2] (x4) [label=below:{\small $0$}] at (0.825,-1) {}; 
  \node [vertex2] (x5) [label=below:{\small $1$}] at (1.35,-1) {};
  \node [vertex2] (x6) [label=below:{\small $0$}] at (1.65,-1) {}; 
  \node [vertex2] (x7) [label=below:{\small $1$}] at (2.175,-1) {};
  \node [vertex2] (x8) [label=below:{\small $0$}] at (2.475,-1) {}; 
  \node [vertex2] (x9) [label=below:{\small $1$}] at (3,-1) {};
  \node [vertex2] (x10) [label=below:{\small $0$}] at (3.3,-1) {}; 

  \draw [edge2] (v1) -- (x1); \draw [edge2] (v1) -- (x4); \draw [edge2] (v1) -- (x6); \draw [edge2] (v1) -- (x7); \draw [edge2] (v1) -- (x10);  
  \draw [edge2] (v2) -- (x2); \draw [edge2] (v2) -- (x3); \draw [edge2] (v2) -- (x6); \draw [edge2] (v2) -- (x8); \draw [edge2] (v2) -- (x10);  
  \draw [edge2] (v3) -- (x2); \draw [edge2] (v3) -- (x3); \draw [edge2] (v3) -- (x6); \draw [edge2] (v3) -- (x7); \draw [edge2] (v3) -- (x9);
  \draw [edge] (v4) -- (x2); \draw [edge] (v4) -- (x3); \draw [edge] (v4) -- (x5); \draw [edge] (v4) -- (x7); \draw [edge] (v4) -- (x9);  
 \end{tikzpicture} \hspace{1cm}
 \begin{tikzpicture}
  \tikzstyle{vertex}=[circle,draw=black,fill=lightgray,minimum size=6pt,inner sep=0pt, outer sep=0pt]
  \tikzstyle{vertex2}=[circle,draw=black,fill=lightgray,minimum size=4pt,inner sep=0pt, outer sep=0pt]
  \tikzstyle{edge}=[draw=black]
  \tikzstyle{edge2}=[draw=lightgray]
 
  \node at (1.6, 0.9) {\textbf{Bob}};
  \node at (-0.8, -1) {$B_2$};
  \node [vertex] (v1) [label={\small $a_1$}] at (0,0) {};
  \node [vertex] (v2) [label={\small $a_2$}] at (1,0) {};  
  \node [vertex] (v3) [label={\small $a_3$}] at (2,0) {};  
  \node [vertex] (v4) [label={\small $a_4$}] at (3,0) {};   
  \node [vertex2] (x1) [label=below:{\small $1$}] at (-0.3,-1) {};
  \node [vertex2] (x2) [label=below:{\small $0$}] at (-0,-1) {};   
  \node [vertex2] (x3) [label=below:{\small $1$}] at (0.525,-1) {};
  \node [vertex2] (x4) [label=below:{\small $0$}] at (0.825,-1) {}; 
  \node [vertex2] (x5) [label=below:{\small $1$}] at (1.35,-1) {};
  \node [vertex2] (x6) [label=below:{\small $0$}] at (1.65,-1) {}; 
  \node [vertex2] (x7) [label=below:{\small $1$}] at (2.175,-1) {};
  \node [vertex2] (x8) [label=below:{\small $0$}] at (2.475,-1) {}; 
  \node [vertex2] (x9) [label=below:{\small $1$}] at (3,-1) {};
  \node [vertex2] (x10) [label=below:{\small $0$}] at (3.3,-1) {}; 

  \draw [edge2] (v1) -- (x1); \draw [edge2] (v1) -- (x3); \draw [edge2] (v1) -- (x6); \draw [edge2] (v1) -- (x7); \draw [edge2] (v1) -- (x9);  
  \draw [edge] (v4) -- (x2); \draw [edge] (v4) -- (x3); \draw [edge] (v4) -- (x6); \draw [edge] (v4) -- (x7); \draw [edge] (v4) -- (x10);  
 \end{tikzpicture} \hspace{1cm}
 \begin{tikzpicture}
  \tikzstyle{vertex}=[circle,draw=black,fill=lightgray,minimum size=6pt,inner sep=0pt, outer sep=0pt]
  \tikzstyle{vertex2}=[circle,draw=black,fill=lightgray,minimum size=4pt,inner sep=0pt, outer sep=0pt]
  \tikzstyle{edge}=[draw=black]
  \tikzstyle{edge2}=[draw=lightgray]
 
  \node at (1.6, 0.9) {\textbf{Charlie}};
  \node at (-0.8, -1) {$B_3$};
  \node [vertex] (v1) [label={\small $a_1$}] at (0,0) {};
  \node [vertex] (v2) [label={\small $a_2$}] at (1,0) {};  
  \node [vertex] (v3) [label={\small $a_3$}] at (2,0) {};  
  \node [vertex] (v4) [label={\small $a_4$}] at (3,0) {};   
  \node [vertex2] (x1) [label=below:{\small $1$}] at (-0.3,-1) {};
  \node [vertex2] (x2) [label=below:{\small $0$}] at (-0,-1) {};   
  \node [vertex2] (x3) [label=below:{\small $1$}] at (0.525,-1) {};
  \node [vertex2] (x4) [label=below:{\small $0$}] at (0.825,-1) {}; 
  \node [vertex2] (x5) [label=below:{\small $1$}] at (1.35,-1) {};
  \node [vertex2] (x6) [label=below:{\small $0$}] at (1.65,-1) {}; 
  \node [vertex2] (x7) [label=below:{\small $1$}] at (2.175,-1) {};
  \node [vertex2] (x8) [label=below:{\small $0$}] at (2.475,-1) {}; 
  \node [vertex2] (x9) [label=below:{\small $1$}] at (3,-1) {};
  \node [vertex2] (x10) [label=below:{\small $0$}] at (3.3,-1) {}; 

  \draw [edge] (v4) -- (x2); \draw [edge] (v4) -- (x4); \draw [edge] (v4) -- (x6); \draw [edge] (v4) -- (x7); \draw [edge] (v4) -- (x9);  
 \end{tikzpicture}
 \caption{In the example instance given in Figure~\ref{fig:bit-vector-learning-main}, Alice holds  
 $Y_1^1 = 10010$, $Y_1^2 = 01000$, $Y_1^3 = 01011$, and $Y_1^4 = 01111$. For each string $Y_1^j$, Alice connects vertex
 $a_j$ to $5$ vertices, each indicating one bit of the respective bit string. For example, when reading the labels
 of the $B_1$-vertices connected to $a_4$ from left-to-right, we obtain the bit sequence $01111$ which equals $Y_1^4$.
 \label{fig:reduction-main}}
 \end{center}
\end{figure*}

\begin{theorem} \label{thm:insertion-only-lower-bound}
 Let $\mathbf{A}$ be an $\alpha$-approximation streaming algorithm for \textsf{FEwW} with error probability at most $0.005$ 
 and $\alpha = \frac{p}{1.01}$, for some integer $p \ge 2$. Then $\mathbf{A}$ uses space at least:
 $$ \Omega(\frac{dn^{\frac{1}{p-1}}}{\alpha^2}) \ .$$ 
 %there exists a communication protocol 
 %$\mathbf{P}$ for \textsf{Bit-Vector Learning}$(\alpha, n, k)$, for any value of $k$, where each message is of length at most $s(n,2k \cdot \alpha)$.
\end{theorem}
\begin{proof}
%Let $\mathbf{A}$ be an algorithm as in the statement of the theorem. We will show how 
%$\mathbf{A}$ can be used to solve an instance of \textsf{Bit-Vector Learning}$(p, n, k)$.

Given their inputs for \textsf{Bit-Vector Learning}$(p, n, k)$, the $p$ parties construct a graph 
$$G=([n], [2 k p], \cup_{i=1}^p {E_i})$$ 
so that party $i$ holds edges $E_i$. The edges of party $i \in [p]$ are as follows:
 $$E_i = \{ (\ell,2 k \cdot (i-1) + 2 \cdot (j-1) + Y_i^{\ell}[j] + 1) \ : \ \ell \in X_{i}  \mbox{ and }  j \in [k] \} \ . $$
 An illustration of this construction is given in Figure~\ref{fig:reduction-main} (the example uses the notation $B_i = \{2k(p-1)+1, \dots, 2kp \}$).
 Observe that $\Delta = kp$ (the vertex in $X_p$ has such a degree). 
 
 Let $\mathbf{A}$ be an $\alpha$-approximation streaming algorithm for \textsf{FEwW}$(n, d)$
 with $\alpha=\frac{p}{1.01}$ and $d=\Delta=kp$.
 Party $1$ simulates algorithm $\mathbf{A}$ on their edges $E_1$
 and sends the resulting memory state to party $2$. This continues until party $p$ completes the algorithm and outputs
 a neighbourhood $(I, S)$. We observe that every neighbour $s \in S$ of vertex $I$ allows us to determine one bit of string $Z^I$. Since 
 the approximation factor of $\mathbf{A}$ is $\frac{p}{1.01}$, we have $|S| \ge \frac{1.01 \cdot \Delta}{p} = 1.01k$. We can thus 
 predict $1.01k$ bits of string $Z^I$. By Theorem~\ref{thm:bit-vector-learning}, every such protocol requires a message of length 
 $$\Omega(\frac{kn^{\frac{1}{p-1}}}{p}) = \Omega(\frac{dn^{\frac{1}{p-1}}}{\alpha^2}) \ ,$$
 which implies the same space lower bound for $\mathbf{A}$.
\end{proof}

\section{Upper Bound for Insertion-deletion Streams} \label{sec:alg-insertion-deletion}
In this section, we discuss our streaming algorithm for \textsf{FEwW} for insertion-deletion streams.

Our algorithm is based on the combination of two sampling strategies which both rely on the 
very common {\em $l_0$-sampling} technique: An $l_0$-sampler in insertion-deletion streams outputs a uniform random element
from the non-zero coordinates of the vector described by the input stream.
In our setting, the input vector is of dimension $n \cdot m$ where each coordinate indicates the presence or absence of an edge. 
Jowhari et al. showed that there is an $l_0$-sampler that uses space 
$O(\log^2(dim) \log \frac{1}{\delta})$, where $dim$ is the dimension of the input vector, and succeeds with probability $1-\delta$ \cite{jst11}. 

In the following,
we will run $\OrderT(n d)$ $l_0$-samplers. To ensure that they succeed with large enough probability, we will
run those samplers with $\delta = \frac{1}{n^{10}d}$ which yields a space requirement of $O(\log^2(nm) \cdot \log(nd))$ for each sampler. 

$l_0$-sampling allows us to, for example, sample uniformly at random from all edges of the input graph or from all edges incident to a specific vertex.

Our algorithm is as follows: %We run the following two strategies simultaneously (parametrized by an integer $x$): 

\vspace{0.2cm}

\begin{center}
 \noindent \fbox{ \begin{minipage}{0.45 \textwidth}
\begin{enumerate}
 \item Let $x = \max \{\frac{n}{\alpha}, \sqrt{n} \}$ \vspace{0.1cm}
 \item \textbf{Vertex Sampling:} Before processing the stream, sample a uniform random subset 
 $A' \subseteq A$ of size $10 x \ln n$. For each sampled vertex $a$, run $10 \frac{d}{\alpha} \ln n$ 
 $l_0$-samplers on the set of edges incident to $a$. This strategy requires space $\OrderT(\frac{xd}{\alpha})$.
 \vspace{0.1cm}
 \item \textbf{Edge Sampling:} Run  $10 \frac{nd}{\alpha} \left( \frac{1}{x} + \frac{1}{\alpha} \right) \ln(n m)$ $l_0$-samplers on 
 the stream, each producing a uniform random edge. This strategy requires space $\OrderT \left(\frac{nd}{\alpha} \left( \frac{1}{x} + \frac{1}{\alpha} \right) \right) $.
 \vspace{0.1cm}
 \item Output any neighbourhood of size at least $\frac{d}{\alpha}$ among the stored edges if there is one, otherwise report \texttt{fail}
\end{enumerate}
\end{minipage}
} \\ \nopagebreak
\vspace{0.15cm}
\textbf{Algorithm 3:} One-pass streaming algorithm for insertion-deletion streams
\end{center}

\vspace{0.2cm}

The analysis of our algorithm relies on the following lemma, whose proof uses standard concentration bounds and 
is deferred to the appendix.

\begin{lemma}\label{lem:sampling}
 Let $y, k, n$ be integers with $y \le k \le n$. Let $\mathcal{U}$ be a universe of size $n$ 
 and let $X \subseteq \mathcal{U}$ be a subset of size $k$. 
 Further, let $Y$ be the subset of $\mathcal{U}$ obtained by sampling 
 $C \ln(n) \frac{ny}{k}$ times from $\mathcal{U}$ uniformly at random (with repetition), for some $C \ge 4$. Then, 
 $|Y \cap X| \ge y$ with probability $1 - \frac{1}{n^{C-3}}$.
\end{lemma}

We will first show that if the input graph contains enough vertices of degree at least $\frac{d}{\alpha}$, then the 
vertex sampling strategy succeeds.

\begin{lemma}\label{lem:vertex-sampling}
 The vertex sampling strategy succeeds with high probability if there are at least $\frac{n}{x}$ 
 vertices of degree at least $\frac{d}{\alpha}$.
\end{lemma}
\begin{proof}
 First, we show that $A'$ contains a vertex of degree at least $\frac{d}{\alpha}$ with high probability. Indeed,
 the probability that no node of degree at least $\frac{d}{\alpha}$ is contained in the sample $A'$ is at most:
 \begin{align*}
 \frac{{n - \frac{n}{x} \choose 10 x \ln n}}{{n \choose 10 x \ln n}} & = \frac{(n - \frac{n}{x})! \cdot (n  - 10 x \ln n)!}{n! \cdot (n - \frac{n}{x} - 10 x \ln n)! }  
% & = & \frac{(n  - C_2 x \log n)(n  - C_2 x \log n - 1) \dots (n  - C_2 x \log n - \frac{n}{x} + 1) }{n(n-1) \dots (n - \frac{n}{x} + 1)} \\
 \le \left( \frac{n  - 10 x \ln n}{n} \right)^{\frac{n}{x}}  \\%\\ = \left( 1 - \frac{C_2 x \log n}{n} \right)^{\frac{n}{x}} \\
 & \le  \exp \left( - \frac{10 x \ln n}{n} \cdot \frac{n}{x} \right)  = n^{-10} \ .
 \end{align*}
 Next, suppose that there is a node $a \in A'$ with $\deg(a) \ge \frac{d}{\alpha}$. Then, by Lemma~\ref{lem:sampling}
 sampling $10 \cdot \frac{d}{\alpha} \log n$ times uniformly at random from the set of edges incident to $a$ results
 in at least $\frac{d}{\alpha}$ different edges with probability at least $1- n^{-7}$. 
\end{proof}

Next, we will show that if the vertex sampling strategy fails, then the edge sampling strategy succeeds.
\begin{lemma}\label{lem:edge-sampling}
 The edge sampling strategy succeeds with high probability if there are at most $\frac{n}{x}$ vertices of degree at least $\frac{d}{\alpha}$. 
\end{lemma}
\begin{proof}
Let $\Delta$ be the largest degree of an $A$-vertex. Since there are at most $\frac{n}{x}$ $A$-vertices of degree
at least $\frac{d}{\alpha}$, the input graph has at most $|E| \le \frac{n}{x} \cdot \Delta + n \cdot \frac{d}{\alpha}$ edges.
%By Lemma~\ref{lem:sampling}, the edge sampling strategy samples at least $\frac{xd}{\alpha}$ different edges from the 
Fix now a node $a$ of degree $\Delta$. Then, by Lemma~\ref{lem:sampling}, we will sample $\frac{d}{\alpha}$ different edges 
incident to $a$ with high probability, if we sample
\begin{align*}
10 \cdot \frac{|E| \frac{d}{\alpha}}{\Delta} \ln (|E|) & \le 10 \cdot \left( \frac{nd}{x\alpha} + \frac{nd^2}{\alpha^2\Delta} \right) \ln(|E|) \\
& \le 10 \cdot \frac{nd}{\alpha} \left( \frac{1}{x} + \frac{1}{\alpha} \right) \ln(n \cdot m) 
\end{align*}
times, which matches the number of samples we take in our algorithm.
%uniform random edges, where we used $\ln(m) \le \ln(n^2) = 2 \ln n$ and $\Delta \ge d$.
\end{proof}

We obtain the following theorem:
\begin{theorem} \label{thm:insertion-deletion-upper-bound}
 Algorithm~3 is a one-pass $\alpha$-approximation streaming for insertion-deletion streams that uses space 
 $\OrderT(\frac{dn}{\alpha^2})$ if $\alpha \le \sqrt{n}$, and space $\OrderT(\frac{\sqrt{n}d}{\alpha})$ if $\alpha > \sqrt{n}$, and succeeds with high probability.
\end{theorem}
\begin{proof}
 Correctness of the algorithm follows from Lemmas~\ref{lem:vertex-sampling} and \ref{lem:edge-sampling}. 
 Concerning the space requirements, the algorithm uses space $\OrderT(\frac{xd}{\alpha}) + \OrderT(\frac{nd}{\alpha}(\frac{1}{x} + \frac{1}{\alpha}))$,
 which simplifies to the bounds claimed in the statement of the theorem by choosing $x = \max \{\frac{n}{\alpha}, \sqrt{n} \}$.
\end{proof}
Using the same ideas as in the proof of Corollary~\ref{cor:star-detection}, we obtain:
\begin{corollary}\label{cor:insertion-deletion-star-detection}
 There is a $\Order(\sqrt{n})$-approximation semi-streaming algorithm for insertion-deletion streams for \textsf{Star Detection}
 that succeeds with high probability.
\end{corollary}

\section{Lower Bound for Insertion-deletion Streams} \label{sec:lb-insertion-deletion}
We will give now our lower bound for \textsf{FEwW} in insertion-deletion streams. To this end,
we first define the two-party communication problem \textsf{Augmented-Matrix-Row-Index} and 
prove a lower bound on its communication complexity. Finally,
we argue that an insertion-deletion streaming algorithm for \textsf{FEwW} can be used to 
solve \textsf{Augmented-Matrix-Row-Index}, which yields the desired lower bound.

\subsection{The \textsf{Augmented-Matrix-Row-Index} Problem} \label{sec:aug-matrix-row-index}
Before defining the problem of interest, we require additional notation. Let $M$ be an $n$-by-$m$ matrix. 
Then the $i$th row of $M$ is denoted by $M_i$. A position $(i,j)$ is a tuple chosen from $[n] \times [m]$. 
We will index the matrix $M$ by a set of positions $S$, i.e., $M_S$, meaning the matrix positions
$M_{i,j}$, for every $(i,j) \in S$.

The problem \textsf{Augmented-Matrix-Row-Index}$(n,m,k)$ is defined as follows:
\begin{problem}[\textsf{Augmented-Matrix-Row-Index}$(n,m,k)$]\label{prob:aug-matrix-row-index}
In the problem \textsf{Augmented-Matrix-Row-Index}, Alice holds 
a binary matrix $X \in \{0, 1 \}^{n \times m}$ where every $X_{ij}$ is a uniform random Bernoulli variable, 
for some integers $n,m$. Bob holds a uniform random index 
$J \in [n]$ and for each $i \neq J$, Bob holds a uniform random subset of positions $Y_i \subseteq \{ i \} \times [m]$ with 
$|Y_i| = m - k$ and also knows $X_{Y_i}$. 
Alice sends a message to Bob who then outputs the entire row $X_J$.
\end{problem}
For ease of notation, we define $Y_I = \bot$ and $Y=Y_1, Y_2, \dots, Y_n$. An example instance of 
\textsf{Augmented-Matrix-Row-Index}$(4,6,2)$ is given in Figure~\ref{fig:augmented-matrix-row-index-example-main}.

\begin{figure}[H]
\begin{center}
 \begin{tikzpicture}
        \matrix [matrix of math nodes,left delimiter=(,right delimiter=)] (x) at (0,0)
        {
            0 & 1 & 1 & 1 & 0 & 0 \\               
            1 & 1 & 0 & 0 & 1 & 0 \\               
            0 & 0 & 0 & 0 & 1 & 0 \\
            1 & 0 & 1 & 0 & 1 & 0 \\
        };
        
        \matrix [matrix of math nodes,left delimiter=(,right delimiter=)] (m) at (5,0)
        {
            0 & 1 & 1 &  &  & 0 \\               
            1 & 1 &  & 0 & 1 &  \\               
            {\color{white}0} &  &  &  &  & {\color{white}0} \\
             & 0 & 1 & 0 &  & 0 \\
        };
        \draw[color=black, pattern=dots] (m-3-1.north west) -- (m-3-6.north east) -- (m-3-6.south east) -- (m-3-1.south west) -- (m-3-1.north west);
        %\node (txt) at (8, -0.3) {row $J$};
        
        \node (A) at (0, 1.5) {\textbf{Alice}};
        \node (B) at (5, 1.5) {\textbf{Bob}};
        
        \node(C) at (1.7,1.2) {};
        \node(D) at (3.3,1.2) {};
	\node(E) at (2.5,1.5) {$M$};
        \draw[->] (C) -- (D);
    \end{tikzpicture}

 \caption{Example Instance of \textsf{Augmented-Matrix-Row-Index}$(4,6,2)$. Bob needs to output the content of row $3$.
 Bob knows $6-2=4$ random positions in every row except row $3$.
 \label{fig:augmented-matrix-row-index-example-main} }
 \end{center}
\end{figure}
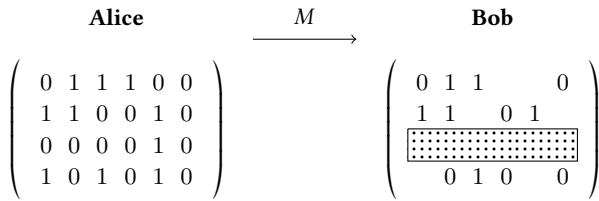

% \begin{figure}
% \begin{center}
%  \begin{tikzpicture}
%         \matrix [matrix of math nodes,left delimiter=(,right delimiter=)] (x) at (0,0)
%         {
%             0 & 1 & 1 & 1 & 0 & 0 \\               
%             1 & 1 & 0 & 0 & 1 & 0 \\               
%             0 & 0 & 0 & 0 & 1 & 0 \\
%             1 & 0 & 1 & 0 & 1 & 0 \\
%         };
%         
%         \matrix [matrix of math nodes,left delimiter=(,right delimiter=)] (m) at (7,0)
%         {
%             0 & 1 & 1 &  &  & 0 \\               
%             1 & 1 &  & 0 & 1 &  \\               
%             {\color{white}0} &  &  &  &  & {\color{white}0} \\
%              & 0 & 1 & 0 &  & 0 \\
%         };
%         \draw[color=black, pattern=dots] (m-3-1.north west) -- (m-3-6.north east) -- (m-3-6.south east) -- (m-3-1.south west) -- (m-3-1.north west);
%         %\node (txt) at (8, -0.3) {row $J$};
%         
%         \node (A) at (0, 1.5) {\textbf{Alice}};
%         \node (B) at (7, 1.5) {\textbf{Bob}};
%         
%         \node(C) at (2.5,1.2) {};
%         \node(D) at (4.5,1.2) {};
% 	\node(E) at (3.5,1.5) {$M$};
%         \draw[->] (C) -- (D);
%     \end{tikzpicture}
% 
%  \caption{Example Instance of \textsf{Augmented-Matrix-Row-Index}$(4,6,2)$. Bob needs to output the content of row $3$.
%  Bob knows $6-2=4$ random positions in every row except row $3$.
%  \label{fig:augmented-matrix-row-index-example} }
%  \end{center}
% \end{figure}

\subsection{Lower Bound Proof for \textsf{Augmented-Matrix-Row-Index}}\label{sec:aug-matrix-row-index-lb}
We now prove a lower bound on the one-way communication complexity of \textsf{Augmented-Matrix-Row-Index}$(n,m,k)$. 
To this end, let $\Pi$ be a deterministic communication protocol for \textsf{Augmented-Matrix-Row-Index}$(n,m,k)$ 
with distributional error at most $\epsilon > 0$ and denote by $M$ the message that Alice sends to Bob. 

First, we prove that the mutual information between row $X_J$ and Bob's knowledge, that is $M J Y X_Y$, is large. 
Since the proof of the next lemma is almost identical to Lemma~\ref{lem:fano-type} we postpone it to 
the appendix:
\begin{lemma} \label{lem:fano-turnstile}
 We have:
 $$I(X_J \ : \ M J Y X_Y) \ge (1-\epsilon)m -  1\ .$$
\end{lemma}

Next, we prove our communication lower bound for \textsf{Augmented-Matrix-Row-Index}:

\begin{theorem} \label{thm:aug-matrix-row-index}
 We have: %The one-way $\epsilon$-error communication complexity of \textsf{Augmented-Matrix-Row-Index}$(n,m,k)$ is at least:
 $$R_\epsilon^{\rightarrow}(\textsf{Augmented-Matrix-Row-Index}(n,m,k)) \ge (n-1) (k - 1 - \epsilon m) \ .$$
\end{theorem}

\begin{proof} 
 Our goal is to bound the term $I(X \ : \ M)$ from below. %, which proves a suitable lower bound since $I(X \ : \ M) \le H(M) \le s$.
 To this end, we partition the matrix $M$ as follows: Let $Z$ be all positions that are different to row $J$ and the positions 
 known to Bob, i.e., the set $Y$. Then:
 \begin{align*}
  I(X \ : \ M) & = I(X_Y X_J X_Z \ : \ M) \\
  & = I(X_Y X_J  \ : \ M) + I(X_Z \ : \ M \ | \ X_J X_Y) \\
  & \ge  I(X_Z \ : \ M \ | \ X_J X_Y) \ ,
 \end{align*}
 where we applied the chain rule for mutual information. For $i \neq J$, let $Z_i = (\{i\} \times [m]) \setminus Y_i$, i.e.,
 the positions of row $i$ unknown to Bob, and let $Z_J = \varnothing$. Furthermore, let $L$ be a random variable that is 
 uniformly distributed in $[n] \setminus J$. Consider now a fixed index $j$. Then, using the chain rule for mutual information and 
 the fact that $X_{Z_i}$ and $X_{Z_q}$ are independent, for every $i \neq q$, we obtain:
\begin{align*}
 I(X_Z \ : \ M \ | \ X_j X_Y, J = j) %& = & %\sum_{i \in [n] \setminus j} I(X_{Z_i} \ : \ M | \ X_J X_Y X_{Z_1}, \dots, X_{Z_{i-1}}) \\
 & \ge \sum_{i \in [n] \setminus \{j \} } I(X_{Z_i} \ : \ M | \ X_j X_Y, J = j) \ , \\
 & = (n-1) \cdot I(X_{Z_L} \ : \ M | \ X_J X_Y L, J = j) \ .
\end{align*}
By combining all potential values for $j$, we obtain:
$$I(X_Z \ : \ M \ | \ X_J X_Y) \ge (n-1) \cdot I(X_{Z_L} \ : \ M | \ X_J X_Y L) \ . $$

In the following, we will show that $I(X_{Z_L} \ : \ M | \ X_J X_Y L) \ge k - 1 - \epsilon m$, which then completes the theorem. To this end, we will relate
the previous expression to the statement in Lemma~\ref{lem:fano-turnstile}, as follows: First, let $Y'_J$ be $m-k$ uniform random 
positions in row $J$. Then by independence, we obtain
\begin{eqnarray*}
 I(X_{Z_L} \ : \ M | \ X_J X_Y L) \ge I(X_{Z_L} \ : \ M | \ X_{Y'_J} X_Y L) \ .
\end{eqnarray*}
Next, denote by $Y \setminus Y_L := Y_1, \dots, Y_{L-1}, Y_{L+1}, \dots, Y_n$. Then, by using the chain rule again, we obtain:
\begin{align*}
 I(X_L \ : \ M | \ X_{Y_J'} X_{Y \setminus Y_L} L) & =  I(X_{Y_L} X_{Z_L}  \ : \ M | \ X_{Y_J'} X_{Y \setminus Y_L} L) \\
 & =  I(X_{Y_L} \ : \ M | \ X_{Y_J'} X_{Y \setminus Y_L} L) \\
 & \quad \quad + I(X_{Z_L} \ : \ M | \ X_{Y_J'} X_{Y} L) \\
 & \le H(X_{Y_L}) + I(X_{Z_L} \ : \ M | \ X_{Y_J'} X_{Y} L) \\
 & \le  (m-k) + I(X_{Z_L} \ : \ M | \ X_{Y_J'} X_{Y} L) \ .
\end{align*}
Last, it remains to argue that $I(X_{Z_L} \ : \ M | \ X_{Y_J'} X_{Y \setminus Y_L} L)$ is equivalent to 
$I(X_{Z_J} \ : \ M | \ J Y X_Y)$. Indeed, first observe that $L$ is chosen uniformly at random from $[n] \setminus J$, which is equivalent
to a value chosen uniformly at random from $[n]$ since $J$ is itself a uniform random value in $[n]$. Observe further that
the conditioning is also equivalent: both $X_{Y_J'} X_{Y \setminus Y_L}$ and $X_Y$ reveal $m-k$ uniform random positions 
of each row different to row $L$ and $J$, respectively. Hence, using Lemma~\ref{lem:fano-turnstile} we obtain:
\begin{align*}
I(X_{Z_L} \ : \ M | \ X_{Y_J'} X_{Y} L) & \ge I(X_L \ : \ M | \ X_{Y_J'} X_{Y \setminus Y_L} L) - (m-k) \\
& \ge (1-\epsilon)m - 1 - (m-k) = k - 1 - \epsilon m \ .
\end{align*}
We have thus shown that $I(X \ : \ M) \ge (n-1) (k - 1 - \epsilon m)$. The result then follows, since $I(X \ : \ M) \le H(M) \le |M|$.
\end{proof}

\subsection{Reduction: \textsf{FEwW} to \textsf{Augmented-Matrix-Row-Index}}
\label{sec:reduction-aug-matrix-neighbourhood-detection}

\begin{lemma} \label{lem:reduction-turnstile}
 Let \textbf{A} be an $\alpha$-approximation insertion-deletion streaming algorithm for \textsf{FEwW}$(n,d)$
 with space $s$ that fails with probability at most $\delta$. Then there is a one-way communication protocol for
 $\textsf{Augmented-Matrix-Row-Index}(n,2d,\frac{d}{\alpha} - 1)$ with message size $$O( s \cdot \alpha \cdot \log n)$$
 that fails with probability at most $\delta + n^{-10}$.
\end{lemma}
\begin{proof}
 We will show how algorithm \textbf{A} can be used to solve $\textsf{Augmented-Matrix-Row-Index}(n,2d,\frac{d}{\alpha}-1)$. Assume from now on that the number of $1$s in row $J$ of matrix $X$ is at least $d$.
 We will argue later what to do if this is not the case. Alice and Bob repeat the following protocol $\Theta(\alpha \log n)$ times
 in parallel:
 
 First, Alice and Bob use public randomness to chose $n$ permutations $\pi_{i}: [2d] \rightarrow [2d]$ 
 at random and permute the elements of each row $i$ independently using $\pi_i$. 
 Observe that this 
 operation does not change the number of $1$s in each row. Let $X'$ be the permuted matrix. 
 Then, Alice and Bob interpret the matrix 
 $X'$ as the adjacency matrix of a bipartite graph, where Bob's knowledge about $X'$ is treated as edge 
 deletions. Under the assumption that row $J$ contains at least $d$ $1$s, and since none of the elements
 of row $J$ are deleted by Bob's input, we have a valid instance for \textsf{FEwW}$(n,d)$. 
 Alice then runs \textbf{A} on the graph obtained from $X'$
 and sends the resulting memory state to Bob. Bob then continues \textbf{A} on his input and outputs
 a neighbourhood of size at least $\frac{d}{\alpha}$. Observe that after Bob's deletions, every row except row $J$ 
 contains at most $\frac{d}{\alpha} - 1$ $1$s, which implies that \textbf{A} reports a neighbourhood rooted at 
 $A$-vertex $J$ (the vertex that corresponds to row $J$). Bob thus learns at least $\frac{d}{\alpha}$ positions of row $J$ where the matrix $X'$ 
 is $1$. Bob then applies $(\pi_J)^{-1}$ and thus learns at least $\frac{d}{\alpha}$ positions of row $J$ of matrix $X$
 where the value is $1$. Observe that since the permutation $\pi_J$ was chosen uniformly at random, the probability
 that a specific position with value $1$ in row $J$ of matrix $X$ is learnt by the algorithm is at least 
 $\frac{d/\alpha}{2d} = \frac{1}{2\alpha}$. Applying concentration bounds, since the protocol is repeated $\Theta(\alpha \cdot \log n)$ times (where $\Theta$
 hides a large enough constant), we learn all $1$s in row $J$ with probability $1-n^{-10}$ and thus have 
 solved $\textsf{Augmented-Matrix-Row-Index}(n,2d,\frac{d}{\alpha}-1)$.
 
 It remains to address the case when row $J$ contains fewer than $d$ $1$s. To address this case, Alice and Bob 
 simultaneously run the algorithm mentioned above on the matrix obtained by inverting every bit, which allows
 them to learn all positions in row $J$ where the matrix $X$ is $0$. Finally, Bob can easily decide in which
 of the two cases they are: If row $J$ contained at most $d-1$ $1$s then the strategy without inverting the input 
 would therefore report at most $d-1$ $1$s. 
\end{proof}

\begin{theorem} \label{thm:insertion-deletion-lower-bound}
 Every $\alpha$-approximation insertion-deletion streaming algorithm for 
 \textsf{FEwW}$(n,d)$ that fails with probability
 $\delta \le \frac{1}{2d}$ requires space $\Omega \left(\frac{nd}{\alpha^2 \log n}) \right)$.
\end{theorem}
\begin{proof}
 Let $\mathbf{A}$ be a streaming algorithm as in the description of this theorem. Then, by Lemma~\ref{lem:reduction-turnstile},
 there is a one-way communication protocol for $\textsf{Augmented-Matrix-Row-Index}(n,2d,\frac{d}{\alpha} - 1)$ that succeeds with
 probability $\delta + n^{-10}$ and communicates $\Order(s \cdot \alpha \log n)$ bits. Then, by Theorem~\ref{thm:aug-matrix-row-index},
 we have:
 $$s \cdot \alpha \log n = \Omega \left((n-1)(\frac{d}{\alpha} - 2 - (\delta + n^{-10}) 2d \right) = \Omega \left(\frac{nd}{\alpha} \right) \ ,$$
 which implies 
 $$ s = \Omega \left( \frac{nd}{\alpha^2 \log n}  \right) \ .$$
\end{proof}

%2-$\epsilon$ approximation is simple, requires space $\Omega(dn)$.

\bibliographystyle{ACM-Reference-Format}
\bibliography{stars}

\appendix

\section{Sampling Lemma}

\noindent \textbf{Lemma~\ref{lem:sampling}.} \textit{
 Let $y, k, n$ be integers with $y \le k \le n$. Let $\mathcal{U}$ be a universe of size $n$ 
 and let $X \subseteq \mathcal{U}$ be a subset of size $k$. 
 Further, let $Y$ be the subset of $\mathcal{U}$ obtained by sampling 
 $C \ln(n) \frac{ny}{k}$ times from $\mathcal{U}$ uniformly at random (with repetition), for some $C \ge 4$. Then, 
 $|Y \cap X| \ge y$ with probability $1 - \frac{1}{n^{C-3}}$. }
\begin{proof}
 Let $t_i$ be the expected number of samples it takes to sample an item from $X$ that has 
 not been sampled previously, given that $i-1$ items from $X$ have already been sampled. The probability 
 of sampling a new item given that $i-1$ items have already been sampled is $p_i = \frac{k - (i-1)}{n}$, 
 which implies that $t_i = \frac{1}{p_i} = \frac{n}{k-(i-1)}$. Thus, the expected number $\mu$ of samples 
 required to sample at least $y$ different items is therefore:
 \begin{eqnarray*}
  \mu := \sum_{i=1}^y t_i = \sum_{i=1}^y \frac{n}{k-(i-1)} = n \cdot (H_k - H_{k-y}) = n \cdot H \ ,
 \end{eqnarray*}
 where $H_i$ is the $i$-th Harmonic number and $H = H_k - H_{k-y}$. We consider two cases:

 Suppose first that $y \ge \frac{k}{2}$. Then, we use the approximation $n \le \mu \le n \ln(k)$. 
 By a Chernoff bound, the probability that more than $C \ln(n) \frac{ny}{k} \ge \frac{C}{2} n \ln(n)$ samples are needed is at most
 $$\exp \left(- \frac{(\frac{C}{2}-1)^2}{2 + \frac{C}{2}-1} n \right) \le \exp \left( - \frac{1}{2} n \right)  \ . $$
 
 Next, suppose that $y < \frac{k}{2}$. Then, we use the (crude) approximations $1 \le \mu \le \frac{ny}{k}$.
 By a Chernoff bound, the probability that more than $C \ln(n) \frac{ny}{k}$ samples are needed is at most:
 $$\exp \left(- \frac{(C-1)^2 \ln(n)^2}{2 + (C-1)\ln n} \right) \le n^{-C+3} \ . $$

\end{proof}

\section{Missing Proof: Insertion-deletion Stream Lower Bound}

%\setcounter{theorem}{\ref{lem:fano-turnstile}}
%\addtocounter{theorem}{-1}

\hspace{0.3cm} \textsc{Lemma \ref{lem:fano-turnstile}} \textit{
 We have:}
 $$I(X_J \ : \ M J Y X_Y) \ge (1-\epsilon)m -  1\ .$$
 
%\end{lemma}
\begin{proof}
Let $Out$ be the output produced by the protocol for \textsf{Augmented-Matrix-Row-Index}. 
 We will first bound the term $I(Out \ : \ X_J) = H(X_J) - H(X_J \ | \ Out)$. To this end, let $E$ be the indicator random 
 variable of the event that the protocol errs. Then, $\Pr [E=1] \le \epsilon$. We have:
 \begin{align} \label{eqn:309}
  H(E, X_J \ | \ Out) = H(X_J \ | \ Out) + H(E \ | \ Out, X_J ) = H(X_J \ | \ Out) \ ,
 \end{align}
 where we used the chain rule for entropy and the observation that $H(E \ | \ Out, X_J ) = 0$ since $E$ is fully determined
 by $Out$ and $X_J$. Furthermore,
 \begin{align} 
  \nonumber H(E, X_J \ | \ Out) & = H(E \ | \ Out) + H(X_J  \ | \ E, Out) \\
  \label{eqn:310} & \le 1 +  H(X_J \ | \ E, Out) \ ,
 \end{align}
using the chain rule for entropy and the bound $H(E \ | \ Out) \le H(E) \le 1$. From Inequalities~\ref{eqn:309} and \ref{eqn:310}
we obtain:
\begin{eqnarray} \label{eqn:330}
 H(X_J \ | \ Out) & \le & 1 +  H(X_J \ | \ E, Out) \ .
\end{eqnarray}
Next, we bound the term $H(X_J \ | \ E, Out)$ as follows:
\begin{align} \nonumber
 H(X_J \ | \ E, Out) & = \Pr \left[ E = 0 \right] H(X_J \ | \ Out, E = 0) \\ 
\label{eqn:311} & + \Pr \left[ E = 1 \right] H(X_J \ | \ Out, E = 1)  \ .
\end{align}
Concerning the term $H(X_J \ | \ Out, E = 0)$, %observe that $Z^I$ is a bit string of length at most 
%$p \cdot k$ and is uniformly distributed. 
since no error occurs, $Out$ determines $X_J$.
We thus have that $H(X_J \ | \ Out, E = 0) = 0$.
We bound the term $H(X_J \ | \ Out, E = 1)$ by $H(X_J \ | \ Out, E = 1) \le H(X_J) = m$ (since conditioning can only decrease entropy).
The quantity $H(X_J \ | \ E, Out)$ can thus be bounded as follows:
\begin{eqnarray} \label{eqn:312}
 H(X_J \ | \ E, Out) & \le & (1-\epsilon) \cdot 0 + \epsilon H(X_J) = \epsilon H(X_J)  \ .
\end{eqnarray}
Next, using Inequalities~\ref{eqn:330} and \ref{eqn:312}, we thus obtain:
\begin{align*}
 I(Out \ : \ X_J) & = H(X_J) - H(X_J \ | \ Out) \\
 & \ge H(X_J) - 1 - H(X_J \ | \ E, Out) \\
 & \ge H(X_J) - 1 - \epsilon H(X_J) \\
 & = (1-\epsilon) H(X_J) - 1 = (1-\epsilon) m - 1 \ .
\end{align*}
Last, observe that $Out$ is a function of $M, J, Y$ and $X_Y$. The result then follows from the data processing inequality.
\end{proof}

\end{document}